\documentclass[11pt]{article}
\usepackage{jheppub}
\usepackage{amsmath,amssymb,amsfonts,amsthm,mathtools,bm,bbm,mathrsfs}
\usepackage[dvipsnames]{xcolor}
\usepackage{hyperref}
\usepackage[utf8]{inputenc}
\usepackage[titletoc]{appendix}
\usepackage{parskip}

\usepackage{tikz}
\usepackage{tikz-3dplot}
\usepackage{tikz-cd}
\usetikzlibrary{snakes}
\usetikzlibrary{decorations.pathmorphing}
\usetikzlibrary{intersections,shapes.geometric}

\makeatletter
\def\@fpheader{\relax}
\makeatother

\newcommand{\D}{\Delta}
\newcommand{\Dp}{\Delta^*}

\newcommand\be{\begin{equation}}
	\newcommand\ee{\end{equation}}
\newcommand\bea{\begin{eqnarray}}
	\newcommand\eea{\end{eqnarray}}
\newcommand\ba{\begin{array}}
	\newcommand\ea{\end{array}}
\newcommand\nn{\nonumber}

\newcommand\bc{\begin{center}}
	\newcommand\ec{\end{center}}
\newcommand\pa{\partial}
\newcommand\comment[1]{}
\renewcommand\tilde{\widetilde}
\newcommand{\Log}{\mathrm{Log}}
\newcommand{\RLC}{\check{C}^+_z}
\newcommand{\rom}[1]{\uppercase\expandafter{\romannumeral #1\relax}}

\newcounter{ex}

\newtheorem{proposition}{Proposition}
\theoremstyle{remark}
\newtheorem{remark}{Remark}
\newtheorem{example}[ex]{Example}{\bf}{\rm}

\numberwithin{theorem}{section}
\numberwithin{proposition}{section}
\numberwithin{remark}{section}
\numberwithin{equation}{section}
\numberwithin{ex}{section}

\allowdisplaybreaks

\newcommand\Z{\mathbb{Z}}
\newcommand\R{\mathbb{R}}
\newcommand\C{\mathbb{C}}
\newcommand\CP{\mathbb{P}}

\title{Tropical Periods for Calabi-Yau Hypersurfaces in non--Fano Toric Varieties}

\author[a]{Per Berglund}
\author[a]{\!\!, Michael Lathwood}

\affiliation[\,a]{Department of Physics and Astronomy, University of New Hampshire, Durham, NH 03824, USA}

\emailAdd{Per.Berglund@unh.edu}
\emailAdd{Michael.Lathwood@unh.edu}

\abstract{
We consider multi-polytopes to describe non-Fano toric varieties and their associated anticanonical Calabi-Yau hypersurfaces. From the periods of the mirror manifold  
the $\widehat{\Gamma}$-conjecture is shown to hold for examples of Calabi-Yau hypersurfaces in non-Fano  ambient spaces, extending  earlier work by Abouzaid et al by employing  a  generalized Duistermaat-Heckman measure.}

\begin{document}
	
	\maketitle
	\parskip=10pt
	
	\section{Introduction and results}\label{sec:intro}
	Calabi-Yau 3-folds play an essential role in string compactifications
 \cite{Candelas1985}. 
	There exists a number of ways to construct these Ricci-flat K\"{a}hler manifolds.
	 One such class is obtained by considering the anticanonical divisor in a toric variety using the data encoded in a 4-dimensional reflexive polytope \cite{Bat93}.
	 Since reflexive polytopes in 4 dimensions are completely classified \cite{KS}, it is of interest to look beyond the reflexive case.
	 We proceed in this direction by presenting a new way to calculate period integrals for Calabi-Yau hypersurfaces in non-Fano toric varieties using the generalized Duistermaat–Heckman measure. 
	 We first review intersection theory and mirror symmetry  for examples of Calabi-Yau hypersurfaces $X$ in toric varieties $\mathcal{Y}_\D$ in Sections \ref{sec:intersec} and \ref{sec:MS}, respectively.
	 In particular, we  use the $m$-twisted Hirzebruch $(n+1)$-folds $\mathcal{F}^{(n+1)}_m$ for $n=1,2,3$ as our running examples.
     Mirror symmetry allows us to use the mirror manifold $\check{X}$ to compute the period of a mirror cycle $\RLC$ in terms of the intersection data of $\mathcal{Y}_\D$.
	 In Section \ref{sec:gDH} we  turn to our main tool for calculating this period: the generalized Duistermaat-Heckman measure $\overline{DH}_{\D,\xi}$.
	 The generalized Duistermaat-Heckman measure is a measure on $\R^{n+1}$ that is  used to calculate the volume of multi-polytopes \cite{HattoriMasuda}, relating it to 
  the symplectic volume of $\mathcal{Y}_\D$ \cite{DHbook}.
	 Following \cite{Nis06}, in Section \ref{sec:gDH1} we construct the generalized Duistermaat-Heckman measure in terms of the intersection of the dual cones $\overline{U}(I)^+$ in $\R^2$.
	 For the higher dimensional case,   we introduce a graded ring $H_*^\text{trop}(\Sigma)$ which allows for a completely algebraic representation of the generalized Duistermaat-Heckman measure, see Section \ref{sec:gDH2}.
	 
	 Using our results, in Section \ref{sec:gDH3} we  compute the period of a cycle $\RLC\hookrightarrow\check{X}$ in the mirror to $X$ using the intersection data of $\mathcal{Y}_\D$ and the Duistermaat-Heckman theorem. This is the content of Proposition \ref{prop:K3period} for $n=2$ and in Proposition \ref{prop:CY3Period} for $n=3$.
	 Abouzaid-Ganatra-Iritani-Sheridan \cite{Sheridan_etal} first proved a period formula of this type in any number of dimensions using a tropical decomposition of the cycle $\RLC$.
	  In this paper, we analyze new cases of K3s and Calabi-Yau 3-folds using an appropriate generalization of their calculation.
      Our calculation allows for multiple complex structure moduli $(z_1,\dots,z_s)$ and toric ambient spaces that are defined by multi-polytopes, which are exactly the generalizations we need for $\mathcal{F}^{(n+1)}_m$.
      We then get explicit formulas in terms of the intersection data of the ambient toric variety $\mathcal{Y}_\D$.
	  
	  Lastly in Section \ref{sec:EIT}, we recalculate the Euler characteristic term of the period via ``error in tropicalization". This method, also inspired by Abouzaid-Ganatra-Iritani-Sheridan, computes the difference between the amoeba of $\RLC$ and $\partial\D$, which is the unique compact connected component of the tropical amoeba of $\RLC$. We obtain new contributions to the period due to properties of the Newton multi-polytopes for our non-Fano examples. Such contributions are novel to these cases and lead us to conjecture the existence of tropical varieties with extension regions $\D_\text{ext}$ as in \cite{BH16}.

   After this article was submitted to the arXiv, it was brought to our attention by Helge Ruddat that there are other methods of computing periods tropically \cite{ruddat2019period}. 
   His work with Bernd Siebert outlines how one can construct tropical cycles $\beta_\text{trop}$ that lift to cycles in a toric degeneration $\mathfrak{X}$.
   The period of this lifted cycle is then a monomial if written in canonical coordinates of the mirror complex structure moduli space.
   One can then relate the antiderivative of this period to the mirror superpotential $W$ \cite{GRZ}.
   However, this procedure has not been explicitly carried out for non-Fano toric varieites.
   We have carried out a similar analysis in these interesting cases which will be the subject of a forthcoming work \cite{wip}.
	
	\section{Intersection theory for Calabi-Yau hypersurfaces in toric varieties}
		\label{sec:intersec}
		As previously stated, a very large class of Calabi-Yau manifolds is given by hypersurfaces in toric varieties. We define the toric ambient spaces now and explain how to compute topological quantities associated to them using polytope data. 
		In particular, we are interested in constructing the anticanonical hypersurface $X$ in a toric variety $\mathcal{Y}_\D$ using its divisor data.
		We then explain how to calculate the Chern class $c(X)$ of $X$ and the Euler characteristic $\chi$.
		The following material is standard, but two good references are \cite{CLS,Fulton}.
        A good reference for this material in the context of string theory is \cite{Denef}.
	\subsection{Toric varieties}
	 Let $M\cong\Z^{n+1}$ be an $(n+1)$-dimensional lattice and let $N=\text{Hom}(M,\Z)\cong\Z^{n+1}$ be the dual lattice. Suppose we are given a lattice polytope 
	\begin{equation}
		\Delta^* \subset N_\R
	\end{equation}
	where $N_\R=N\otimes\R\cong \R^{n+1}$. That is, $\Delta^*$ is a real $(n+1)$-dimensional polyhedron with $n$-dimensional faces (facets). Note that we do not assume that $\Delta^*$ is reflexive. In practice, one first specifies the vertices of $\Delta^*$, which are given by points $v_\rho \in N$. The one-dimensional rays in the normal fan $\Sigma^{(1)}\subset\Delta^*$ of the dual polytope
	\begin{equation}
		\label{polar}
		(\Delta^*)^\circ=\{\mu\in M_\R \,|\, \mu\cdot\nu\ge-1 \,\,\, \forall \nu\in\Delta^*\}\subset M_\R
	\end{equation} 
	consists of the rays that go through the vertices $v_\rho\in\Delta^*$. When $\Delta^*$ is convex, $(\Delta^*)^\circ=\Delta$ is called the \textit{Newton polytope}. In other words, $v_\rho\in\Delta^*$ are the inward pointing normal vectors of the faces of $\Delta$. One can use $\Sigma^{(1)}$  to describe the $(n+1)$-dimensional torus $T$ orbits in an $(n+1)$-dimensional toric variety $\mathcal{Y}_\Delta$. Further, each $v_\rho\in\Sigma^{(1)}$ corresponds to a $T$ invariant divisor $D_\rho \subset \mathcal{Y}_\Delta$, meaning a $n$-dimensional (codimension 1) subspace of $\mathcal{Y}_\Delta$. There is a scheme-theoretic construction that sets $\mathcal{Y}_\Delta=\text{Proj} \, S_\Delta$ where $S_\Delta$ is the polytope ring, but here we are just concerned with the divisor data of the toric variety, which is encoded in the \textit{spanning polytope} $\Delta^*$.
	
	\begin{example}[$\mathcal{Y}_\D=\mathcal{F}^{(3)}_m$]
		\label{ex:Hirzebruch3d}
		Let $n+1=3$. The fan
		\begin{equation}
			\Sigma^{(1)}=\{(-1,-1,0),(1,0,0),(0,1,0),(0,0,1),(-m,-m,-1)\}=\{v_\rho\}_{\rho=1}^5
		\end{equation}
		describes $\mathcal{Y}_\Delta=\mathcal{F}^{(3)}_m$, the  \textit{$m$-twisted Hirzebruch 3-fold}, which is a $\mathbb{P}^2$ fibration over $\mathbb{P}^1$. The first two coordinates of the $v_\rho$ correspond to the $\mathbb{P}^2$ and the last coordinate corresponds to the $\mathbb{P}^1$. Since $\Sigma^{(1)}_{\mathbb{P}^2}=\{(1,0),(0,1),(-1,-1)\}$ and $\Sigma^{(1)}_{\mathbb{P}^1}=\{1,-1\}$, we can interpret $m$ parameter in $v_4$ as controlling the ``twisting" of the bundle. Essentially, $m$ is changing the relationship between the fibers and the base. From the plot of $\Delta^*$ for $m=3$ in Figure \ref{F^3_3 spanning polytope}, we  see that the spanning polytope becomes non-convex for $m>2$. If we want to construct the Newton polytope $\Delta$, we have to use the construction given in \cite{BH16}. But we  soon see how to construct $\Delta$ using the generalized Duistermaat-Heckman measure.
		\begin{figure}[h!]
			\begin{center}
				\includegraphics[scale=0.4]{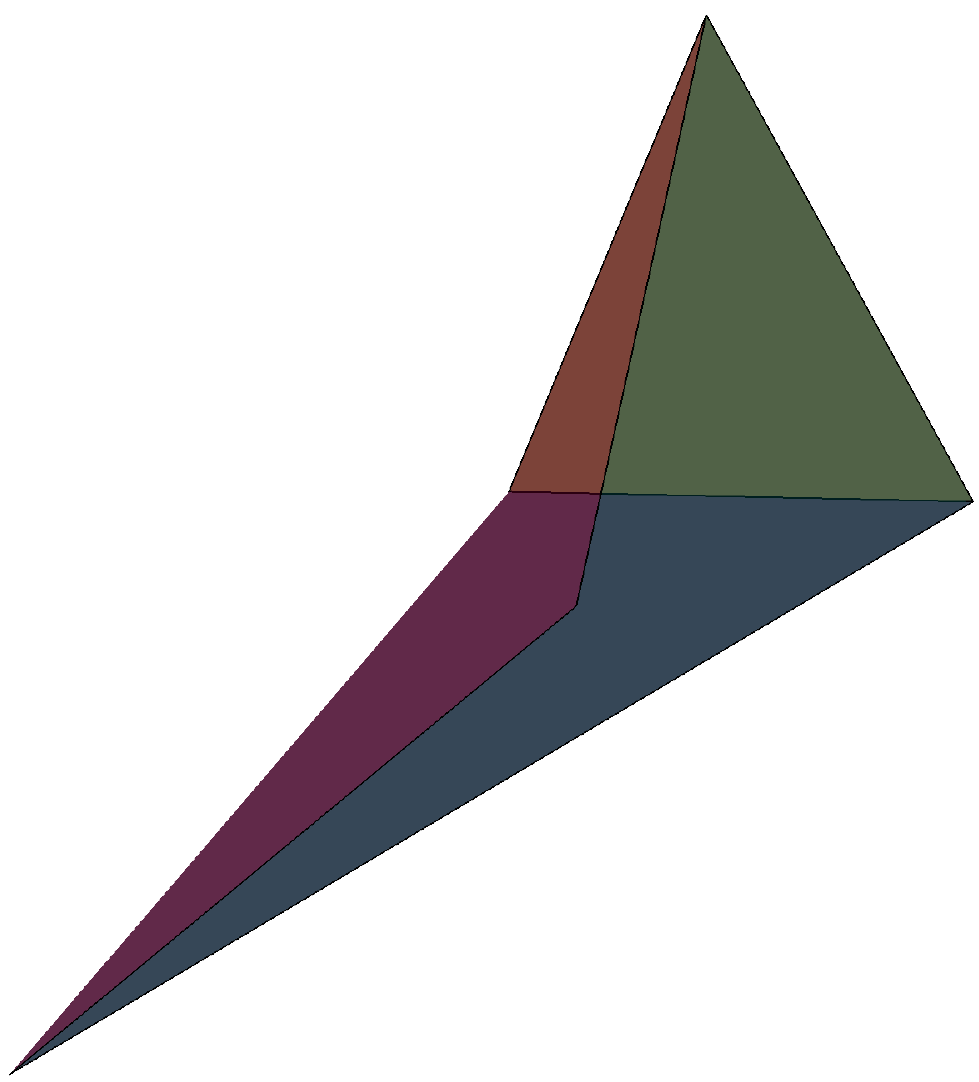}
			\end{center}
			\caption{The 3-dimensional spanning polytope $\Delta^*$ for the 3rd Hirzebruch 3-fold. The ``VEX point" $(-1,-1,0)$ prevents this polytope from being convex when the bundle is twisted to this extent \cite{BH16}.}
			\label{F^3_3 spanning polytope}
		\end{figure}
		\par We can create a ``cone" over $\Delta^*$ by setting $\bar{v}_\rho=(v_\rho,1)$ and $\bar{v}_0=(0,\dots,0,1)$. The linear relations between the $D_\rho$ in the Chow ring\footnote{The Chow ring $A_*(X)$ is a graded ring whose elements are divisor classes modulo linear equivalence and the grading is by dimension. The multiplication on $A_*(X)$ is given by the intersection product, which we define in Equation \ref{intprod}.} $A_{n}(\mathcal{Y}_\Delta)$ can be deduced from additive identities of the $\bar{v}_\rho$. Consider the following additive identities of the $\bar{v}_\rho$.
		\begin{align}
			\label{linrelate}
			\bar{v}_1+\bar{v}_2+\bar{v}_3-3\bar{v}_0=0 \\
			-m\bar{v}_1+\bar{v}_4+\bar{v}_5-(2-m)\bar{v}_0=0
		\end{align}
		The corresponding divisors then satisfy the following identities.
		\begin{equation}
			D_1=D_2-mD_4\,\, , \,\, D_2=D_3 \,\, , \,\, D_4=D_5
		\end{equation}
		Thus the Chow ring is generated as a graded ring by two linearly independent divisor classes 
		\begin{equation}
			A_*(\mathcal{F}^{(3)}_m)=\langle D_2, D_4 \rangle
		\end{equation}
	From now on we write $\{D_1, D_2\}$ as our basis of linearly independent divisor classes, where $D_1$ is the class corresponding to the $\CP^2$ fiber and $D_2$ is the class corresponding to the $\CP^1$ base.
	\end{example}
	So far our discussion of toric varieties has been without reference to string theory, but toric varieties can also be realized as the classical supersymmetric ground states of a gauged linear sigma model (GLSM), modulo the $U(1)^s$ symmetry \cite{Wit93}. If $x_i$ are the scalar components of $k$ chiral superfields with charges $Q_i^a$ under a $U(1)^s$ gauge group, then the potential of the corresponding GLSM is given by
	\begin{equation}
		V(x)=\sum_{a=1}^s\frac{e_a^2}{2}\left(\sum_{i=1}^kQ_i^a|x_i|^2-\xi^a\right)^2,
	\end{equation}
	where the $e_a$ are the $U(1)^s$ coupling constants, and $\xi^a$ are the Fayet-Iliopoulos (FI) parameters. We take the zeros of $V$ modulo the gauge group action (called the D-flat configurations)
	\begin{equation}
		\mathcal{M}=\{x\in\C^{k} \,|\, V(x)=0\}/U(1)^s
	\end{equation}
	where the $U(1)^s$ acts according to the charge matrix
	\begin{equation}
		x_j\mapsto e^{iQ_j^a\varphi_a}x_j.
	\end{equation}
	Then, $\mathcal{M}$ is a $d$-dimensional toric variety with $d:= n+1=k-s$. A toric variety constructed this way also inherits a symplectic form from $\C^k$
	\begin{equation}
		\omega=\frac{i}{2\pi}\sum_j dx^j\wedge d\bar{x}^j \in H^{1,1}(\mathcal{M}).
	\end{equation}
	One can construct a basis $\{C^a\}\subset A_1(\mathcal{M})$ of 2-cycles by taking $d-1$ intersections of the divisors $D_\rho$. In this GLSM construction, these divisors are simply
	\begin{equation}
		D_\rho=\{x_\rho=0\}\in H_{d-1}(\mathcal{M})
	\end{equation}
	With our symplectic form and a basis of 2-cycles, define the \textit{K\"{a}hler moduli} $t_a$
	\begin{equation}
		t_a=\int_{C^a} \, \omega.
	\end{equation}
	These parameters control the ``size" of the toric variety, since they essentially give the length of curves. It follows that the dimension of the K\"{a}hler moduli space $\mathscr{M}_\omega$, i.e. the number of different K\"{a}hler parameters, is the number of $U(1)$ factors $s$ in the GLSM construction.
	\par We gave the GLSM construction here to be utilized in the next example to compute the \textit{intersection product} of $d$ divisors, which can be defined as
	\begin{equation}
		\label{intprod}
		D_{\rho_1}\dots D_{\rho_d}=\int_\mathcal{M} D^*_{\rho_1}\wedge\dots\wedge D^*_{\rho_d}.
	\end{equation}
	Here $D^*_{\rho_i}$ is a $(1,1)$-form that is Poincar\'{e} dual\footnote{The precise definition of this duality is not important here, we just need to know there is a way to associate a $(d-1)$-dimensional divisor with $(1,1)$-form} to $D_{\rho_i}$. It can be shown that the intersection product  $D_{\rho_1}\dots D_{\rho_d}$ is equal to the number of points where all of the $D_{\rho_i}$ intersect. Intersection products of less than $d$ divisors are defined similarly, but now they represent higher dimensional subspaces rather than a number of points. 
	\begin{remark}
			It should be noted that $\mathcal{M}=\mathcal{Y}_\D$ when the rows of the charge matrix $Q_i^a$ are generators of the Mori cone $\ell^{a}$ for $\mathcal{Y}_\D$.
			This will be discussed more in Section \ref{ClosedMS}.
	\end{remark}
	\begin{example}[$\mathcal{F}^{(3)}_m$ GLSM]
		\label{GLSMintersect}
		The Hirzebruch 3-fold $\mathcal{F}^{(3)}_m$ is given by the classical supersymmetric ground states of the GLSM with five chiral superfields and the following $U(1)\times U(1)$ charge matrix
		\begin{equation}
			\begin{bmatrix}
				Q^ 1_i \\
				Q^2_i
			\end{bmatrix} 
			= 
			\left[\begin{tabular}{c c c c c}
				1 &1& 1& 0& 0 \\
				$-m$ & 0 &0 &1 &1 
			\end{tabular} \right]
		\end{equation}
		The D-flatness condition tells us
		\begin{equation}
			\mathcal{F}^{(3)}_m=\bigg\{x\in\C^5 \bigg| \begin{matrix}
				|x_1|^2+|x_2|^2+|x_3|^2=\xi_1 \\
				-m|x_1|^2+|x_4|^2+|x_5|^2=\xi_2	\end{matrix}\bigg\}\bigg/U(1)^2
		\end{equation}
		Since $\xi_a>0$, we have that $(0,0,0,x_4,x_5)\notin\mathcal{F}^{(3)}_m$ and $(x_1,x_2,x_3,0,0)\notin\mathcal{F}^{(3)}_m$, giving us the intersection numbers
		\begin{equation}
			D_1D_2D_3=0 \,\, , \,\, D_4D_5=0.
		\end{equation}
		These relations generate the Stanley-Reisner ideal $\mathcal{I} \subset A_*(\mathcal{F}^{(3)}_m)$. Together with the linear equivalence relations between the $D_\rho$, the relations that generate $\mathcal{I}$ allow us to compute the nontrivial intersection numbers. For example,
		\begin{equation*}
			D_1D_2D_3=D_1^2(D_1-mD_2)=0 \implies D_1^3=m.
		\end{equation*}
		Below we list all the independent triple intersection numbers for later use.
		\begin{equation}
			D_1^3=m \,\, , \,\, D_1^2D_2=1 \,\, , \,\, D_1D_2^2=0 \,\, , \,\, D_2^3=0
		\end{equation}
		The triple intersection numbers can also be written in matrix form.
		\begin{equation}
			C_{1jk}=\begin{bmatrix}
				m & 1 \\
				1 & 0
			\end{bmatrix}
			\,\,\, , \,\,\, 	C_{2jk}=\begin{bmatrix}
				1 & 0 \\
				0 & 0
			\end{bmatrix}
		\end{equation}
		We can see these matrices are symmetric, and we don't have to list the other indices since they  be the same for linearly equivalent divisors. 
	\end{example}
	\subsection{Calabi-Yau hypersurfaces}
	The anticanonical divisor $D_0=-K_{\mathcal{Y}_\Delta}$ of the toric variety $\mathcal{Y}_\Delta$ is given by the sum of the toric divisors
	\begin{equation}
		D_0=\sum_{\rho\in\Sigma^{(1)}} D_\rho
	\end{equation}
	By the adjunction formula, the canonical class of $D_0$ is $0$
	\begin{equation}
		K_{D_0}=(K_{\mathcal{Y}_\Delta} + D_0 )\big|_{D_0}=K_{\mathcal{Y}_\Delta} - K_{\mathcal{Y}_\Delta} =0, 
	\end{equation}
	so an anticanonical section is a Calabi-Yau $n$-fold $X$, since having trivial canonical class is equivalent to having $c_1=0$. The geometric phase of a GLSM with D-flat configurations given by $\mathcal{Y}_\Delta$ is a nonlinear sigma model on $X$ \cite{Denef}. This can be obtained from the GLSM by including Witten's $P$ field and a superpotential $W$. The charges of the $P$ field can be read off from the coefficients of $\bar{v}_0$ in Equation \ref{linrelate}, and we  discuss the superpotential in the next section. We say $X$ is the \textit{anticanonical hypersurface} in $\mathcal{Y}_\D$.  We can compute quantities of interest on $X$ by pulling back to the ambient space.  In practice this is done by inserting a factor of $c_1(\mathcal{Y}_\Delta)$.
	\par For a toric variety $\mathcal{Y}_\Delta$, the total Chern class
	\begin{equation}
		c=1+c_1+\dots+c_d \in H^0(\mathcal{Y}_\Delta)\oplus H^2(\mathcal{Y}_\Delta)\oplus\dots\oplus H^{2d}(\mathcal{Y}_\Delta)
	\end{equation}
	is given by
	\begin{equation}
		\label{ChernToric}
		c(\mathcal{Y}_\Delta):= c(T\mathcal{Y}_\Delta)=\prod_{\rho\in\Sigma^{(1)}}(1+D_\rho),
	\end{equation}
	Here the $D_\rho$ are the cohomology classes Poincar\'{e} dual to the toric divisor homology classes. Since we are taking a product of these classes, it should come as no surprise that $c$ can be computed with intersection theory.

	\begin{example}[$c(X)$ for anti-canonical hypersurface in $\mathcal{F}^{(3)}_m$]
		\label{ChernK3}
		We calculate the Chern class of a K3 $X\hookrightarrow\mathcal{F}^{(3)}_m$ using the adjunction formula and the intersection properties of the toric divisors. We can use Equation \ref{ChernToric} and the linear equivalences of the toric divisors that we derived in Example \ref{ex:Hirzebruch3d} to write
		\begin{equation*}
			c(\mathcal{F}^{(3)}_m)=(1+D_1)^2(1+D_1-mD_2)(1+D_2)^2.
		\end{equation*}
		The normal bundle to the K3 is in the same divisor class as $D_0$, so we get 
		\begin{equation*}
			c(NX)=1+D_0=1+3D_1+(2-m)D_2
		\end{equation*}
		We have a short exact sequence of vector bundles
		\begin{equation}
			0\rightarrow TX \rightarrow T\mathcal{F}^{(3)}_m \rightarrow NX \rightarrow 0
		\end{equation}
		and thus the adjunction formula gives
		\begin{align}
			c(X)&=\frac{c(\mathcal{F}^{(3)}_m)}{c(NX)}\bigg|_X=\frac{(1+D_1)^2(1+D_1-mD_2)(1+D_2)^2}{1+3D_1+(2-m)D_2}\bigg|_X \nonumber \\ 
			&=1+2(3-m)D_1D_2+3D_1^2.
		\end{align}
		To compute this quotient, we Taylor expanded the denominator and dropped terms of degree higher than the complex dimension of $X$. The Euler characteristic of $X$ is given by the integral of the top degree Chern class over $X$. We can pull this computation back to the ambient space by inserting the anticanonical divisor, which by Equation \ref{ChernToric} is the first Chern class of the ambient space.
		\begin{align}
			\label{HirzK3chi}
			\chi(X)&=\int_X c_2(X)=\int_{\mathcal{F}^{(3)}_m} c_1(\mathcal{F}^{(3)}_m)c_2(X) \\
			& =3(6-2m)D_1^2D_4+3(2-m)D_1^2D_4+9D_1^3=24 \nonumber
		\end{align}
		Here we used the intersection numbers derived in Example \ref{GLSMintersect}. This agrees with what we would expect for a K3, which always have $\chi=24$.
	\end{example}

	\section{Mirror symmetry for Calabi-Yau hypersurfaces in toric varieties}
	\label{sec:MS}
		Mirror symmetry asserts that given a Calabi-Yau manifold $X$, there exists a mirror Calabi-Yau manifold $\check{X}$ that yields an equivalent string compactification \cite{GP90,CdlOGP}.
        For a review, see \cite{MS1,MS2}.
		When $\mathcal{Y}_\D$ is semi-Fano, one can construct the mirror to the anticanonical Calabi-Yau hypersurface $X$ from the Landau-Ginzburg superpotential\footnote{We write the exponent vector as $Y^{v_\rho}=Y_1^{v_\rho^1}\dots Y_{n+1}^{v_\rho^{n+1}}$ and include complex coefficients $a_\rho$ which are arbitrary at the moment} $W:(\C^*)^{n+1}\rightarrow\C$ following the Hori-Vafa construction \cite{HoriVafa}
		\begin{align}
			\label{SupPot}
			W(Y_1,\dots,Y_{n+1})=\sum_{\rho\in\Sigma^{(1)}}a_\rho Y^{v_\rho}.
		\end{align}
		In this case, $\check{X}=W^{-1}(0)$. 
		We say that $W$ is the mirror Landau-Ginzburg model to $\mathcal{Y}_\D$, and $W$ can also be used to study mirror symmetry of $X$. 
		Mathematically, we view $W$ as a symplectic fibration over $\C$, and the critical locus of $W$ can be used to determine the quantum cohomology ring $QH(\mathcal{Y}_\D)$ of the ambient toric variety  \cite{FOOO}.
		In particular, for Fano  $\mathcal{Y}_\D$ the quantum cohomology ring is isomorphic to the Jacobian ring of the mirror superpotential.
		\begin{align}
			QH(\mathcal{Y}_\D)\cong\text{Jac}(W)=\C[Y_1,\dots,Y_{n+1}]\bigg/\bigg\langle\frac{\partial W}{\partial Y_i}\bigg\rangle
		\end{align}
		Therefore mirror symmetry of the anticanonical divisor is deeply related to mirror symmetry of the ambient toric variety itself \footnote{The work presented here is for the mirror symmetry of Calabi-Yau pairs, but subsequent work  focus on extending Fano/LG mirror symmetry to the non-Fano cases discussed in the previous section \cite{wip}.}.
		\par There have been studies of mirror superpotentials associated to non-Fano toric varieties \cite{AKO,Aur09,CPS,FOOO}, yet the physical implications of such Landau-Ginzburg models are not fully understood.
		In particular, there are quantum effects due to string worldsheets wrapping cycles in $X$ which are captured by the Gromov-Witten invariants of the space.
		The work presented here, like much of the previous work, is done in the large radius limit of $\mathscr{M}_\omega$ where the worldsheet instanton effects are negligible.
		However, mirror symmetry via period integrals presented in this section is a first step in computing the big quantum cohomology \cite{Iritani11} for examples of non-Fano spaces, and this would give a more complete picture of mirror symmetry in other regions of the moduli space.
		The big quantum cohomology ring of $\mathcal{Y}_\D$ differs from $QH(\mathcal{Y}_\D)$ in that it includes gravitational instanton effects in the ring structure.
		Pioneering work by Fukaya-Oh-Ohta-Ono \cite{FOOO} would suggest there are infinitely many corrections to the superpotential that would be necessary to compute the big quantum cohomology ring.
		However, at least in two dimensions, Auroux has shown that there are only finitely many corrections to $W$ for non-Fano $\mathcal{F}^{(2)}_m$ by preforming an explicit deformation \cite{Aur09}, and his result was reproduced using tropical geometry \cite{CPS}. 
		Such a deformation is possible thanks to a periodicity in the sequence of Hirzebruch $(n+1)$-folds: \be\mathcal{F}^{(n+1)}_m\cong \mathcal{F}^{(n+1)}_{m\, \text{mod} \,(n+1)}\ee
		where the $\cong$ means a diffeomorphism. Although we have this periodicity as real manifolds, it is not known if these diffeomorphic manifolds can be distinguished in quantum cohomology. In this section, we discuss two different perspectives on mirror symmetry with the examples of $\mathcal{F}^{(n+1)}_m$ in mind.
	\subsection{Open-string mirror symmetry}
		Open-string mirror symmetry for compact Calabi-Yau hypersurfaces can be formulated as a quasi-equivalence of $A_\infty$ categories. 
		Following Kontsevich's original proposal \cite{Kon94}, one defines two $A_\infty$ categories $\mathsf{D}_A(X)$ and $\mathsf{D}_B(\check{X})$ to serve as the category of boundary conditions (D-branes) in topologically twisted 2d $N=(2,2)$ superconformal field theories \cite{HIV00}.
		In particular, $\mathsf{D}_A(X)$ is the category of boundary conditions for the $A$-model on $X$ and $\mathsf{D}_B(\check{X})$ is the category of boundary conditions in the $B$-model on $\check{X}$.
		Lagrangian submanifolds $L\in\mathsf{D}_A(X)$ in $X$ serve as the $A$-branes and coherent sheaves $\mathscr{F}\in\mathsf{D}_B(\check{X})$ on $\check{X}$ serve as the $B$-branes \cite{MS2}.
		This means that these objects preserve supersymmetry generated by
		 \begin{align}
		 		 Q_A=\overline{Q}_++Q_- \,\,\,\, \text{and} \,\,\,\, Q_B=\overline{Q}_++\overline{Q}_-
		 \end{align}
		 respectively, where $Q_\pm,\overline{Q}_\pm$ are the supercharges of the 2d $N=(2,2)$ SCFT. 
		 In other words, the physical observables are $Q_A$ or $Q_B$ cohomology classes, which are in one-to-one correspondence with the supersymmetric ground states of the ``open string" boundary theory.
		The cohomological field theory is also manifested in the categorical formulation:
		for any pair D-branes $L, L'\in\mathsf{D}_A(X)$ or $\mathscr{F}, \mathscr{F}'\in\mathsf{D}_B(\check{X})$, the morphism spaces
		 \begin{align}
			\mathsf{D}_A(X)(L,L') \,\,\,\, \text{and} \,\,\,\, \mathsf{D}_B(\check{X})(\mathscr{F},\mathscr{F}')
		\end{align}
		are given by the Hilbert spaces of open string states stretching between the respective D-branes.
			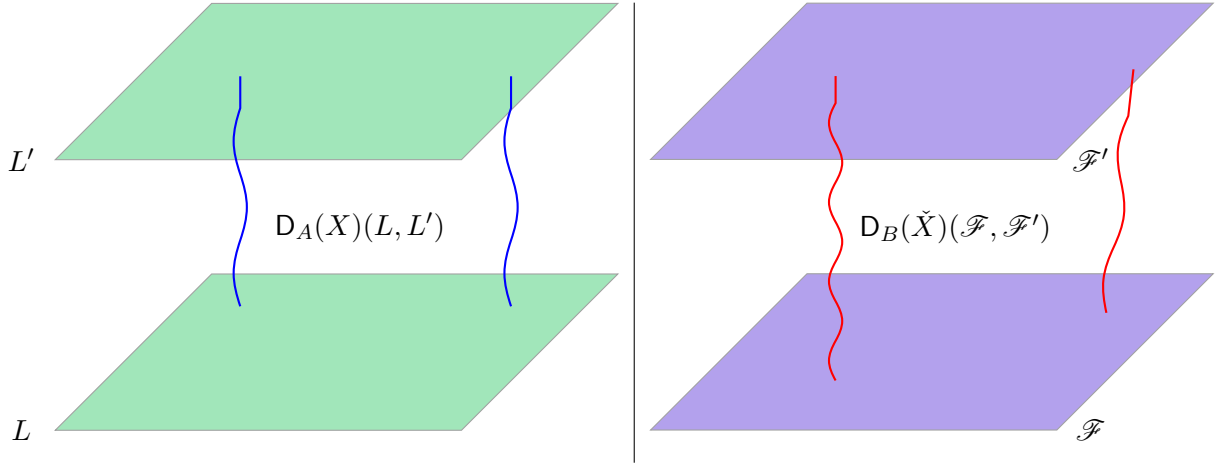
\begin{figure}[h!]
				\begin{center}
					\begin{tabular}{c |c}
							\begin{tikzpicture}[scale=0.9]
							\node at (-3.5,-2,3) {$L$};
							\node at (-3.5,2,3) {$L'$};
							\node at (1.5,1,3) {$\mathsf{D}_A(X)(L,L')$};
							\draw[fill=green,opacity=0.3] (-3,-2,-3) -- (-3,-2,3) -- (3,-2,3) -- (3,-2,-3) -- cycle;
							\draw[fill=blue,opacity=0.1] (-3,-2,-3) -- (-3,-2,3) -- (3,-2,3) -- (3,-2,-3) -- cycle;
							\draw[fill=green,opacity=0.3] (-3,2,-3) -- (-3,2,3) -- (3,2,3) -- (3,2,-3) -- cycle;
							\draw[fill=blue,opacity=0.1] (-3,2,-3) -- (-3,2,3) -- (3,2,3) -- (3,2,-3) -- cycle;
							\draw [thick,blue,snake=coil,segment aspect=0,segment length=50pt] (-2,-1.9,-1.5) -- (-2,1.5,-1.5);
							\draw [thick,blue,snake=coil,segment aspect=0,segment length=50pt] (2,-1.9,-1.5) -- (2,1.5,-1.5);
						\end{tikzpicture}
					&
						\begin{tikzpicture}[scale=0.9]
						\node at (3.5,-2,3) {$\mathscr{F}$};
						\node at (3.5,2,3) {$\mathscr{F}'$};
						\node at (1.5,1,3) {$\mathsf{D}_B(\check{X})(\mathscr{F},\mathscr{F}')$};
						\draw[fill=red,opacity=0.1] (-3,-2,-3) -- (-3,-2,3) -- (3,-2,3) -- (3,-2,-3) -- cycle;
						\draw[fill=blue,opacity=0.3] (-3,-2,-3) -- (-3,-2,3) -- (3,-2,3) -- (3,-2,-3) -- cycle;
						\draw[fill=red,opacity=0.1] (-3,2,-3) -- (-3,2,3) -- (3,2,3) -- (3,2,-3) -- cycle;
						\draw[fill=blue,opacity=0.3] (-3,2,-3) -- (-3,2,3) -- (3,2,3) -- (3,2,-3) -- cycle;
						\draw [thick,red,snake=coil,segment aspect=0,segment length=30pt] (-2,-3,-1.5) -- (-2,1.5,-1.5);
						\draw [thick,red,snake=coil,segment aspect=0,segment length=50pt] (2,-2,-1.5) -- (2.4,1.6,-1.5);
					\end{tikzpicture}
					\end{tabular}
				\end{center}
				\label{fig:Dbrane}
				\caption{A schematic picture of the objects and morphisms in the categories $\mathsf{D}_A(X)$ and $\mathsf{D}_B(\check{X})$}
			\end{figure}
	\par The categorical version of mirror symmetry described above is called homological mirror symmetry. 
	Closed-string mirror symmetry was first observed in the context of 2d $N=(2,2)$ SCFTs \cite{GP90}, and with this new duality it became more tractable to make predictions about counts of rational curves on the quintic threefold \cite{CdlOGP}.
    These works inspired Kontsevich's proposal, as well as a proof of mirror symmetry \cite{Givental,LLY}.
	We can now see homological mirror symmetry as an equivalence of the symplectic topology of $X$ to the complex geometry of $\check{X}$, and vice versa. This can be made more precise in the language of $A_\infty$ categories,
	\begin{align}
		\label{homological mirror symmetry}
		\mathsf{D}_A(X)\simeq\mathsf{D}_B(\check{X}) \,\,\,\, \text{and} \,\,\,\, \mathsf{D}_B(X)\simeq\mathsf{D}_A(\check{X}).
	\end{align}
	The $\simeq$ means a quasi-equivalence, which is an equivalence at the level of cohomology, so one can see mirror symmetry via SCFTs is the same duality as homological mirror symmetry.
    However, the mirror functor gives a finer picture of mirror duality.
	In particular, we can use the mirror functor to match up numerical invariants of objects associated to the mirror pair $X,\check{X}$.
	For a toric variety $\mathcal{Y}_\D$, one would expect a diagram of categorical quasi-equivalences
	\begin{equation}
		\label{diagram}
		\begin{tikzcd}
			\mathsf{D}_B(\mathcal{Y}_\Delta)\arrow[r,"\text{HMS}"] \arrow[d]& \arrow[l] \mathsf{D}_A(M_{\C^*},W) \arrow[r]  & \mathsf{D}_A(\overline{M}_{\C^*},\overline{W}) \arrow[l] \arrow [d] \\
			\mathsf{D}_B(X) \arrow[rr,"\text{HMS}"] \arrow[u,shift right=1ex,hook'] & &\mathsf{D}_A(\check{X}) \arrow[ll] \arrow[u,shift right=1ex,hook']
		\end{tikzcd}
	\end{equation}
	The bottom line of the diagram is Equation \ref{homological mirror symmetry}, which one would expect should be derivable from the categorical data of the ambient space given in the top line.
	However, the proofs of homological mirror symmetry used to construct this diagram for Calabi-Yau hypersurfaces in toric varieties \cite{ProjHMS,GPHMS} all require the ambient toric variety $\mathcal{Y}_\D$ to be Fano\footnote{A detailed discussion of these proofs is outside the scope of this paper since the $A_\infty$ categories change when we look beyond the hypersurface and consider the entire toric variety $\mathcal{Y}_\D$ or its Landau-Ginzburg mirror $W$ and its compactification $\overline{W}$.}.
	\par From the diagram \ref{diagram}, we  show that a particular numerical invariant of objects in $\mathsf{D}_B(X)$ and $\mathsf{D}_A(\check{X})$ matches up as it should for Calabi-Yau hypersurfaces in non-Fano $\mathcal{Y}_\D$.
	On the $B$-side, we  consider the structure sheaf in the derived category of coherent sheaves $\mathscr{O}_X\in\mathsf{D}_B(X)=D^b\text{Coh}(X)$, and on the $A$-side we  consider the real Lagrangian cycle in the Fukaya category $\check{C}^+_z\in\mathsf{D}_A(\check{X})=\mathcal{F}(\check{X})$.
	Abouzaid-Ganatra-Iritani-Sheridan \cite{Sheridan_etal} showed that 
	\begin{align}
	\label{GammaEqn}
	\int_X \widehat{\Gamma}_X\,\,e^{\sum_a t_a\omega_a}=\int_{\check{C}^+_z} \check{\Omega}_z
	\end{align}
	in the large radius limit of moduli space. On the $A$-side,
	\begin{align}
		\widehat{\Gamma}_X=\prod_i \Gamma(1+\delta_i)=\exp\left(-\gamma c_1 +\sum_{k=2}^n (-1)^{k}\zeta(k)(k-1)!\text{ch}_k(X)\right) \in H^*(X;\R)
	\end{align}
	is the $\widehat{\Gamma}$-class,  $c(TX)=\prod(1+\delta_i)$, $\gamma$ is the Euler–Mascheroni constant, $\zeta(k)=\sum_{n=1}^{\infty}\frac{1}{n^k}$ is the Riemann zeta function, and 
	\begin{align}
		\text{ch}(X)=\sum_{k=0}^n \text{ch}_k(X) = n+c_1+\frac{1}{2}\left(c_1^2-2c_2\right)+\frac{1}{6}(c_1^3-3c_1c_2+3c_3)+\dots
	\end{align}
	are the Chern characters of $X$. We can see this data is sensitive to the K\"{a}hler moduli $t_a$ of $X$, which are the expansion coefficients on a basis of K\"{a}hler forms $\{\omega_a\}\subset H^{1,1}(X)$.
	On the $B$-side, we have the period of the real Lagrangian cycle,
	\begin{align}
		\check{C}^+_z=\check{X}_z\cap\R_+^n\cong S^n.
	\end{align}
	with
	\begin{align}
		\check{\Omega}_z=\frac{\prod_i d\log Y_i}{dW_z(Y)}
	\end{align}
	which is dependent on the complex structure moduli $z=(z_1,\dots,z_s)$.
	 Period integrals are the subject of closed-string mirror symmetry, which we  now discuss in order to show Equation \ref{GammaEqn} still holds for $X$ in non-Fano $\mathcal{Y}_\D$.
	
	\subsection{Closed-string mirror symmetry} \label{ClosedMS}
		We now turn to the closed string sector. First, we review one work that utilizes closed-string mirror symmetry \cite{Klemm} to compute periods of the mirror manifold following Batyrev's construction based on the duality of reflexive polytopes \cite{Bat93}. We then discuss work to generalize this construction for certain multi-polytopes \cite{BH16}. 
	\par A Calabi-Yau $n$-fold $X$ is guaranteed to have a holomorphic top form $\Omega \in H^{n,0}(X)$, which can be used to measure the volume of a half-dimensional cycle $C\in H_n(X;\R)$. This volume is called the cycle's \textit{period} $\pi(C)$.
	\begin{align}
		\pi(C):=\int_C \Omega
	\end{align}
	These period integrals satisfy a system of differential equations called the \textit{Picard-Fuchs equations}. We can write this as
	\begin{align}
		\label{PF}
		L_j\pi_i=0
	\end{align}
	where $L_j$ are the Picard-Fuchs differential operators and $\pi_i$ are the periods ($j\in\{1,\dots,s\}\, , \, i\in\{0,\dots,n\}$).  
	In \cite{Klemm}, the Picard-Fuchs equations were derived for complete intersection Calabi-Yau hypersurfaces using the data encoded in $\Delta^*$ following Batyrev's construction \cite{Bat93}. It was found that periods can be expanded in powers of the complex structure moduli $z_a$ of $\check{X}$. 
	These parameters determine the ``shape" of the space, whereas the K\"{a}hler parameters determine the ``size". 
	Under the mirror map, it was shown that one can also express the period in terms of powers of the K\"{a}hler parameters $t_a$ of the original space $X$. 
	If $\mathscr{M}_\omega(X)$ is the K\"{a}hler moduli space of $X$ and $\mathscr{M}_\text{CS}(\check{X})$ is the complex structure moduli space of $\check{X}$, we have
	\begin{align}
		\text{dim}\,\mathscr{M}_\omega(X) = \text{dim}\, \mathscr{M}_\text{CS}(\check{X})
	\end{align} 
	This was first discovered for $n=3$, where $\text{dim}\,\mathscr{M}_\omega(X)=h^{1,1}(X)$ and  $\text{dim}\, \mathscr{M}_\text{CS}(\check{X})=h^{2,1}(\check{X})$.
	Therefore, there are as many K\"{a}hler moduli of $X$ as there are complex structure moduli of $\check{X}$. 
	In the large radius/large complex structure limit, the $t_a$ and $z_a$ are related via
	\begin{align}
		\label{params}
		t_a=-\log z_a.
	\end{align}
	That is, \ref{params} is valid when all $t_a\rightarrow \infty$ and all $z_a \rightarrow 0$, making instanton corrections to the mirror map negligible. 
	In this limit, \cite{Klemm} shows using the Frobenius method that the coefficients on the $(-\log z_a)^0$ and $(-\log z_a)^n$ terms in $\pi_n$ can be calculated using the classical intersection data in Section \ref{sec:intersec}. In particular, the $(-\log z_a)^n$ term is proportional to the triple intersection numbers $C_{ijk}$ when $n=3$, and the $(-\log z_a)^0$ constant term is proportional to the Euler characteristic $\chi$.
	\par For a single defining equation, we can find all the periods $\pi_i$ by taking an appropriate linear combination of derivatives of the following \textit{fundamental period}
		\begin{align}
			\pi_0=\int_{C_0}\frac{a_0}{W_z}\Omega_0,
		\end{align}
	  where $W_z$ is the mirror superpotential given by Equation \ref{SupPot}, except we use the multiplicative relations between the monomials of $W$ encoded in the Mori vectors $\ell^a$ to express the mirror superpotential in terms of the complex structure moduli rather than the coefficients $a_\rho$. The holomorphic volume form on $(\C^*)^{n+1}$,
	  \begin{align}
	  		\Omega_0=d\log Y_1 \wedge \dots \wedge d\log Y_{n+1}=\frac{dY_1}{Y_1}\wedge\dots\wedge\frac{dY_{n+1}}{Y_{n+1}}
	  \end{align}
  	reduces the integral over the unit sphere
  	\begin{align}
  		C_0=\{Y\in (\C^*)^{n+1} \,\, \big| \,\, |Y_1|=\dots=|Y_{n+1}|=1 \}
  	\end{align}
  to a residue integral.
  The residues are given by the constant terms of the power series expansion of $a_0/W_z$.
  The simplicity of this calculation comes from the complex coordinates $Y_i$ in the denominator of $\Omega$. 
  The constant terms of the power series expansion\footnote{By the constant terms, we mean constant with respect to the complex coordinates $Y_i$} can be proportional to $z_a^k$ for all $k\ge 0$, and by using the multinomial theorem, one can calculate the number $c(k_1,\dots,k_s)$ of terms of the from $z_1^{k_1}\dots z_s^{k_s}$ to be
  \begin{align}
  	\label{cMori}
  	c(k_1,\dots,k_s)=\frac{\left(-\sum_a \ell^a_0 k_a\right)!}{\prod_i\left(\sum_a \ell^a_i k_a \right)!}=\frac{\Gamma\left(-\sum_a \ell^a_0 k_a+1\right)}{\prod_i\Gamma\left(\sum_a \ell^a_i k_a+1 \right)}.
  \end{align}
	The Mori vectors $\ell^a$ span the Mori cone of the mirror $\check{X}$, but there is also an interpretation in terms of the original manifold $X$ \cite{Bat95}.
	Namely, treating the linearly independent divisors classes as formal variables, we can write
	\begin{align}
			c(D_1,\dots,D_s)=\frac{\Gamma\left(-K_{\mathcal{Y}_\D}+1\right)}{\prod_{\rho}\Gamma\left(D_\rho+1 \right)},
	\end{align}
	where $-K_{\mathcal{Y}_\D}$ is the anticanonical divisor of the toric ambient space.
	The product in the denominator is over all the toric divisors, and we must choose a basis to express these in terms of the linearly independent divisors $D_1,\dots,D_s$. 
	The procedure outlined above uses the monomial-divisor mirror map to compute the $\pi_i$.
	That said, we  keep the calculation on the mirror side for now.
	 Expanding around the large complex structure limit point $z=0$ of $\mathscr{M}_\text{CS}(\check{X})$, we obtain
	\begin{align}
		\pi_0(z_1,\dots,z_s)=\sum_{k_1,\dots,k_s \ge 0} c(k_1,\dots,k_s)z_1^{k_1}\dots z_s^{k_s}.
	\end{align}
	The derivatives of the $\Gamma$ functions in $c(k_1,\dots,k_s)$ are what give the $\zeta(k)$ coefficients of the higher periods discussed in Equation \ref{GammaEqn} for $k\le n$.
	The leading order terms of the top period $\pi_n$ and the period of the real Lagrangian cycle are equal.
	That is, in the large radius/large complex structure limit we have
	\begin{align}
		\pi_n=\pi(\check{C}^+_z).
	\end{align}
	However it is not known if the terms which are higher order in the $z_a$ can be determined from $\pi(\check{C}^+_z)$. If this could be done, then one could determine the Picard-Fuchs differential operators $L_j$ in Equation \ref{PF} and extend mirror symmetry to the entire moduli space, not just one particular limit point. 
    There are several problems one runs into by proceeding in this direction, so we illustrate a possible issue with an explicit computation.
	\begin{example}[$\pi_0$ for $\mathcal{F}^{(3)}_m$]
		The spanning polytope for the $m$-twisted Hirzebruch 3-fold was given in Example \ref{ex:Hirzebruch3d}, so we could write down its mirror superpotential $W_z$. 
		But we already have the Mori vectors as given by the rows of the GLSM charge matrix in Example \ref{GLSMintersect},
		 \begin{align}
		 	\ell^a=(Q^a_1,\dots,Q^a_r)
		 \end{align}
		 where $r$ is the number of vertices of $\Dp$.
		 Using Equation \ref{cMori}, we can write down the fundamental period for any $m$.
		\begin{align}
			\label{HirzFunPeriod}
			\pi_0=\sum_{k_1,k_2\ge 0}\frac{\Gamma(3k_1+(2-m)k_2+1)}{\Gamma(k_1-mk_2+1)\Gamma^2(k_1+1)\Gamma^2(k_2+1)}z_1^{k_1}z_2^{k_2}
		\end{align}
	For certain values of $(k_1,k_2)$, there  be poles in the $\Gamma$ functions of the denominator which  not be canceled, leading to non-physical infinities. 
	However, we can treat the $k_i$ as continuous parameters and calculate the top period
	\begin{align}
		\pi_2=\frac{1}{2}K_{ij}\pa_{k_i}\pa_{k_j}\pi_0
	\end{align}
	using the intersection numbers 
	\begin{align}
		K_{ij}=K^*D_iD_j,
	\end{align}
	where $D_i$ are the linearly independent toric divisors and $K^*:=-K_{\mathcal{Y}_\D}$. The leading order terms are given by evaluating the derivatives at $k_1=k_2=0$, so we avoid the poles of $\Gamma$ functions. 
	\end{example}
	
	\par Issues arise in the non-Fano cases, like for $\mathcal{F}^{(3)}_m$ when $m\ge 3$, due to non-reflexive spanning polytopes $\Delta^*$. 
	Instead, one can use the \textit{trans-polar} operation $^\nabla$ in \cite{BH16} which generalizes to certain non-reflexive polytopes like $\Delta^*$ for $\mathcal{F}^{(3)}_3$. 
	With this, one can obtain a Newton polytope $\Delta=(\Delta^*)^\nabla$ such that $\Delta^\nabla=\Delta^*$, so that these non-reflexive polytopes exhibit a duality similar to that of reflexive polytpes under the polar operation, $\D^\circ=\Dp$ and $(\Dp)^\circ=\D$.
	\begin{figure}[h!]
		\begin{center}
			\includegraphics[scale=0.5]{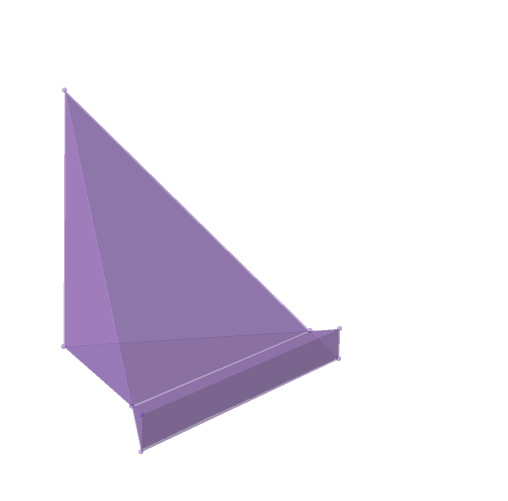}
		\end{center}
		\caption{The Newton polytope $\Delta$ for $\mathcal{F}^{(3)}_3$, constructed as $\Delta=(\Delta^*)^\nabla$}
		\label{3dF3Newton}
	\end{figure}
	In the next section we  present one of the main results of the paper. 
	We  obtain  $\Delta=(\Delta^*)^\nabla$, plotted in Figure \ref{3dF3Newton} for $\mathcal{F}^{(3)}_3$, via an explicit construction using tropical hyperplanes.
	Then, we  compute the leading order terms of $\pi_n$ by computing the period of $\check{C}^+_z$.
	This (closed-string) numerical  invariant gives evidence that the (open-string) homological mirror symmetry results in the previous section can be extended to non-Fano ambient spaces.
	However, for a mirror Landau-Ginzburg model $W$ constructed from a spanning polytope for a non-Fano $\mathcal{Y}_\D$, it is still not known how to modify $\check{X}=W^{-1}(0)$ to get the true mirror to $X$.

\section{Generalized Duistermaat–Heckman measure}
\label{sec:gDH}
We follow the approach of \cite{Nis06} to calculate the generalized Duistermaat–Heckman measure $\overline{DH}_{\D,\xi}$ for our examples of multi-polytopes $\D$.
Namely, we  construct the multi-fan $\Sigma$ for the non-Fano Hirzebruch surface $\mathcal{F}^{(2)}_3$ and the non-Fano Hirzebruch 3-fold $\mathcal{F}^{(3)}_3$, then the DH measure can be written as follows.
\begin{equation}
	\label{gDH}
	\overline{DH}_{\D,\xi}=\sum_{I\in\Sigma^{(n+1)}}(-1)^{|E_I|}w_I\mathbbm{1}_{\overline{U}(I)^+}
\end{equation}
We now unpack this definition to later calculate the period of the real Lagrangian cycle.
\subsection{Multi-fans and multi-polytopes}
\label{sec:gDH1}
Multi-fans and multi-polytopes are generalizations of the fans and polytopes that we have discussed for toric varieties \cite{HattoriMasuda}. 
A multi-fan allows non-trivial winding of the 1-dimensional cones $\Sigma^{(1)}$, and a multi-polytope is a collection of hyperplanes that are dual to a multi-fan.
For our purposes, these generalizations allow top-dimensional cones $I$ to have weights $w_I$ which may be negative due to the winding direction. 
For example, in Figure \ref{fig:F3Multifan} the cone spanned by $\{v_4,v_1\}$ is wound clockwise whereas the other three cones are wound counterclockwise. 
This  be taken into account in Example \ref{MirrorHizmulti-fan}, where the corresponding cone  have negative weight. 
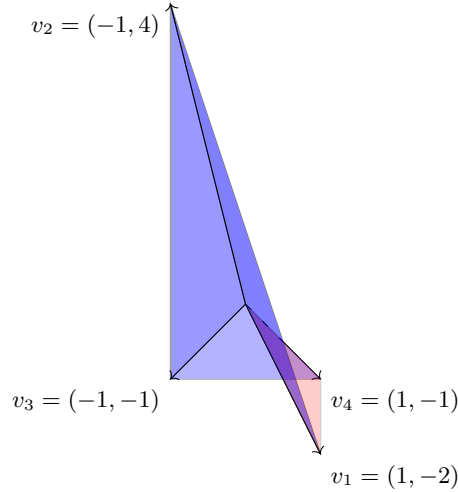
\begin{figure}[h!]
    \centering
    	\begin{tikzpicture}
			\draw[->] (0,0)coordinate(O)--(-1, -1) coordinate(A);
            \draw[loosely dashed] (O) -- (0.5,-0.5)coordinate(I);
			\draw[->] (I) -- (1, -1)coordinate(B);
			\draw[->] (O) -- (1,-2)coordinate(C);
			\draw[->] (O) -- (-1, 4) coordinate(D);
            \draw[below, pos=1, opacity=0.3, fill=blue] (O) -- (A) -- (B) -- cycle;
            \draw[opacity=0.5, fill=blue] (O) -- (D) -- (C) -- cycle;
            \draw[below, pos=2, opacity=0.4, fill=blue] (O) -- (A) -- (D) -- cycle;
            \draw[below, pos=3, opacity=0.2, fill=red] (O) -- (B) -- (C) -- cycle;
			\node[anchor=north east] at (A) {\footnotesize$v_3=(-1,-1)$};
            \node[anchor=north west] at (B) {\footnotesize$v_4=(1,-1)$};
			\node[anchor=north west] at (C) {\footnotesize$v_1=(1,-2)$};
			\node[anchor=north east] at (D) {\footnotesize$v_2=(-1,4)$};
		\end{tikzpicture}
    \caption{The multi-fan $\Sigma=\{\emptyset,\{1\},\{2\},\{3\},\{4\},\{1,2\},\{2,3\},\{3,4\},\{4,1\}\}$ for $\Dp_{\mathcal{F}^{(2)}_3}$. The $\{1,2\}$ cone is ``above" the $\{3,4\}$ cone and the $\{4,1\}$ cone.}
    \label{fig:F3Multifan}
\end{figure}

Let $\Sigma$ be the normal multi-fan of an $(n+1)$-dimensional multi-polytope $\D\subset M_\R\cong\R^{n+1}$.
Similar to the normal fan that we have discussed, the one-dimensional cones $\Sigma^{(1)}$ are generated by the normal vectors $v_i\in N_\R$ to the facets $F_i$ of $\D$.
We write the subscript of the $v_i$ to denote the corresponding cone $\{i\}\in\Sigma^{(1)}$.
Then the $k$-dimensional cones are generated by $k$ of the $v_i$, which we denote by $\{i_1,\dots,i_k\}\in\Sigma^{(k)}$.
The top dimensional cones  be denoted by $I\in\Sigma^{(n+1)}$ and  play a special role as we  sum over these cones to compute the generalized Duistermaat-Heckman measure. 
Each of these cones $I$  be assigned an integer weight $w_I$, which we  take to be $\pm 1$ according to the orientation of the corresponding cone. 
For each $I\in\Sigma^{(n+1)}$, we can define a dual basis $\{u^I_i\}_{i\in I} \subset M_\R$ by imposing 
\begin{equation}
	\langle u^I_i,v_j \rangle = \delta_{ij} \,\,\,\, \forall i,j\in I
\end{equation}
where $\delta_{ij}$ is the Kronecker delta. Pick any $\xi\in N_\R$ such that $\langle u^I_i,\xi \rangle \neq 0$ for any $u_i^I$ and define the subsets
\begin{equation}
	E_I=\{i\in I \,|\,\langle u^I_i,\xi \rangle > 0 \}\subset I
\end{equation}
Finally, define the dual cones 
\begin{equation}
	\overline{U}(I)^+=\bigg\{ u_I+\sum_{ i \in I}r_iu^I_i \, \bigg| \,	\begin{matrix}
		r_i \ge 0 \, ; \, i\in E_I \\
		r_i \le 0 \, ; \, i\notin E_I	\end{matrix} \bigg\} \subset M_\R,
\end{equation}
where $u_I=\cap_{i\in I} F_i$, and the associated characteristic functions of $\mu\in M_\R$
\begin{align}
	\mathbbm{1}_{\overline{U}(I)^+}(\mu)=\begin{cases}
		1 & \mu \in \overline{U}(I)^+ \\
		0 & \mu \notin \overline{U}(I)^+
	\end{cases}.
\end{align}
\par We now illustrate these definitions with three 2-dimensional examples. 
We are most interested in 4-dimensional toric varieties with 3-dimensional Calabi-Yau hypersurfaces, but we first present these cases since they are easier to visualize.
The first example is a Fano toric variety, but then we move on to non-Fano examples.
Essentially we  start with the data of the spanning polytope $\Dp$ and obtain the Newton multi-polytope $\D=(\Dp)^\nabla$, realizing the trans-polar construction through a concrete procedure.
\begin{example}[$\CP^2$ multi-fan]
	\label{P2miltifan}
	For projective space $\CP^n$ (a Fano toric variety), it is not necessary to use the full generality of the multi-fan formalism, but we present it here to show what changes are necessary when moving to a more general case.  Let
	\begin{align}
		v_1&=(1,0)  & v_2&=(0,1) & v_3&=(-1,-1) \\
		I_1&= \{1,2 \} &  I_2&= \{2,3\} & I_3&= \{3,1\}\nonumber\\
		&& \Sigma=\{\emptyset,\{1\},&\{2\},\{3\},I_1,I_2,I_3\} & &
	\end{align}
	The spanning polytope can be constructed as the convex hull
	\begin{equation}
		\Dp=\text{Conv}\{v_1,v_2,v_3\}.
	\end{equation}
	Writing $w_j=w_{I_j}$, we have that
	\begin{equation}
		w_1=w_2=w_3=1,
	\end{equation}
	since each cone has the same orientation. The generalized Duistermaat-Heckman measure can the be written down after calculating the dual bases $\{u^{I_j}_i\}$ to obtain the factors of $(-1)^{|E_{I_j}|}$. Substituting the data for $\CP^2$ into Equation \ref{gDH}, we obtain
	\begin{align}
		\overline{DH}_{\D_{\CP^2},\xi}=\mathbbm{1}_{\overline{U}(I_1)^+}+\mathbbm{1}_{\overline{U}(I_2)^+}-\mathbbm{1}_{\overline{U}(I_3)^+}
	\end{align}
for $\xi=(2,1)$.
\begin{figure}[h!]
	\begin{center}
		\begin{tikzpicture}
			\draw[->,blue] (-1, -1) coordinate(A)-- (-1, 3)coordinate(B);
			\draw[->,blue] (A) -- (4, -1)coordinate(C);
			\draw[->,blue] (2, -1) coordinate(D) -- (C);
			\draw[->,blue] (D) -- (4,-3)coordinate(E);
			\draw[->,red,dashed] (-1,2)coordinate(F) -- (B);
			\draw[->,red,dashed] (F) -- (E);
			\draw[opacity=0.2,fill=red] (F) -- (B) -- (4,3)coordinate(G) -- (E) --cycle;
			\draw[opacity=0.2,fill=blue] (A) -- (B) -- (G) -- (C) --cycle;
			\draw[opacity=0.2,fill=blue] (D) -- (C) -- (E) --cycle;
			\node[anchor=north east] at (A) {\footnotesize$(-1,-1)$};
			\node[anchor=south west] at ([xshift=0.5cm,yshift=0.5cm]A) {\footnotesize$\overline{U}(I_1)^+$};
			\node at (A) {$\bullet$};
			\node[anchor=north east] at (F) {\footnotesize$(-1,2)$};
			\node[anchor=west] at ([xshift=1.5cm,yshift=-0.5cm]F) {\footnotesize$\overline{U}(I_3)^+$};
			\node at (F) {$\bullet$};
			\node[anchor=north east] at (D) {\footnotesize$(2,-1)$};
			\node[anchor=north west] at ([xshift=0.8cm,yshift=-0.4cm]D) {\footnotesize$\overline{U}(I_2)^+$};
			\node at (D) {$\bullet$};
		\end{tikzpicture}
	\end{center}
\caption{The regions $\overline{U}(I_1)^+,\overline{U}(I_2)^+,\overline{U}(I_3)^+\subset M_\R\cong\R^2$ for the multi-fan of $\CP^2$ with top-dimensional cones $\Sigma^{(2)}=\{I_1,I_2,I_3\}$. The positive regions in $\overline{DH}_{\D_{\CP^2},\xi}$ are plotted in blue and the negative in red so that the purple regions are completely canceled.}
\label{P2gDH}
\end{figure}
In Figure \ref{P2gDH}, we can see that the support of the generalized Duistermaat-Heckman measure is exactly the Newton polytope for $\CP^2$. The negative region $\overline{U}(I_3)^+$ completely cancels out $\overline{U}(I_2)^+$, but the positive quadrant given by $\overline{U}(I_1)^+$ is not fully canceled. The remaining part of $\overline{U}(I_1)^+$ is given by
\begin{align}
	\text{supp}\overline{DH}_{\D_{\CP^2},\xi}=\text{Conv}\{(-1,-1),(-1,2),(2,-1)\}=\text{Conv}\{u_{I_j}\}_{j=1}^3=\D_{\CP^2}.
\end{align}
\end{example}
 
 \begin{example}[$\mathcal{F}^{(2)}_3$ multi-fan]
 The spanning polytope for the $m$-twisted Hirzebruch surface $\mathcal{F}^{(2)}_m$ has vertices
 \begin{align}
 	v_1=(-1,0) \,\, , \,\, v_2=(1,0) \,\, , \,\, v_3= (0,1) \,\, , \,\, v_4=(-m,-1)
 \end{align} 
This is plotted for $m=0,1,2,3$ in Figure \ref{fig:HirzSurfPlots}.
 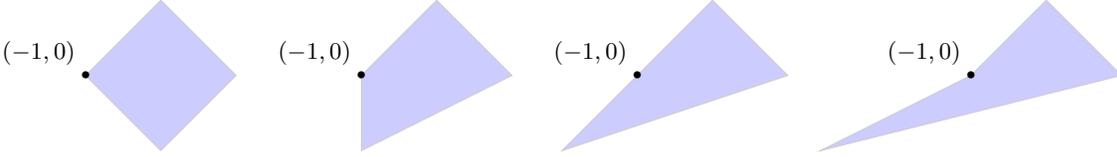
\begin{figure}[h!]
 	\begin{center}
 		\hspace{-1cm}
 		\begin{tabular}{c c c c}
 		\begin{tikzpicture}
 			\draw[opacity=0.2,fill=blue] (-1,0)coordinate(A) -- (0,1)coordinate(B) -- (1,0)coordinate(C) -- (0,-1)coordinate(D) --cycle;
 			\node at (A) {\tiny$\bullet$};
 			\node[anchor=south east] at (A) {\footnotesize$(-1,0)$};
 		\end{tikzpicture}
 	&
 		\begin{tikzpicture}
 		\draw[opacity=0.2,fill=blue] (-1,0)coordinate(A) -- (0,1)coordinate(B) -- (1,0)coordinate(C) -- (-1,-1)coordinate(D) --cycle;
 		\node at (A) {\tiny$\bullet$};
 		\node[anchor=south east] at (A) {\footnotesize$(-1,0)$};
 	\end{tikzpicture}
 &	
 \begin{tikzpicture}
 	\draw[opacity=0.2,fill=blue] (-1,0)coordinate(A) -- (0,1)coordinate(B) -- (1,0)coordinate(C) -- (-2,-1)coordinate(D) --cycle;
 	\node at (A) {\tiny$\bullet$};
 	\node[anchor=south east] at (A) {\footnotesize$(-1,0)$};
 \end{tikzpicture}
&
	\begin{tikzpicture}
	\draw[opacity=0.2,fill=blue] (-1,0)coordinate(A) -- (0,1)coordinate(B) -- (1,0)coordinate(C) -- (-3,-1)coordinate(D) --cycle;
	\node at (A) {\tiny$\bullet$};
	\node[anchor=south east] at (A) {\footnotesize$(-1,0)$};
\end{tikzpicture}
 \end{tabular}
 	\end{center}
 \caption{The spanning polytope $\Dp$ for the $m$-twisted Hirzebruch surface $\mathcal{F}^{(2)}_m$ with $m=0,1,2,3$.}
  \label{fig:HirzSurfPlots}
 \end{figure}
The first vertex $v_1$, called the VEX point \cite{BH16}, makes $\Dp$ non-convex for $m>3$. 
The VEX point for $\Dp_{\mathcal{F}^{(n+1)}_m}$ is $(-1,-1,\dots,-1,0)$ where there are $n$ nonzero entries. We graphically compute the generalized Duistermaat-Heckman measure here for $m=3$, but this can be done for any $m$. Define the top dimensional cones
\begin{align}
	I_1=\{2,3\} \,\, , \,\, I_2=\{1,3\} \,\, , \,\, I_3=\{1,4\} \,\, , \,\, I_4=\{2,4\}
\end{align}
so that our multi-fan is given by $\Sigma^{(2)}=\{I_j\}_{j=1}^4$. Again, all cones have the same orientation, so we can write down the generalized Duistermaat-Heckman measure
\begin{align}
	\overline{DH}_{\D_{\mathcal{F}^{(2)}_3},\xi}=\mathbbm{1}_{\overline{U}(I_1)^+}-\mathbbm{1}_{\overline{U}(I_2)^+}-\mathbbm{1}_{\overline{U}(I_3)^+}+\mathbbm{1}_{\overline{U}(I_4)^+},
\end{align}
where we chose $\xi=(1,1)$. 
\begin{figure}[h!]
	\begin{center}
		\begin{tikzpicture}
			\draw[->,blue] (-1, -1) coordinate(A)-- (-1, 5)coordinate(B);
			\draw[->,blue] (A) -- (2, -1)coordinate(C);
			\draw[->,red, dashed] (1, -1) coordinate(D) -- (C);
			\draw[->,red, dashed] (D) -- (1,5)coordinate(E);
			\draw[->,red,dashed] (1,-2)coordinate(F) -- (E);
			\draw[->,red,dashed] (F) -- (-1.333, 5)coordinate(J);
			\draw[->,blue] (-1, 4) coordinate(H)-- (-1, 5)coordinate(I);
			\draw[->,blue] (H) -- (J);
			\draw[opacity=0.2,fill=red] (F) -- (J) -- (E) -- cycle;
			\draw[opacity=0.2,fill=blue] (A) -- (B) -- (2,5)coordinate(G) -- (C) --cycle;
			\draw[opacity=0.2,fill=red] (D) -- (E) -- (G) -- ([xshift=1cm]D) --cycle;
			\draw[opacity=0.2,fill=blue] (H) -- (I) -- (J) -- cycle;
			\node[anchor=north east] at (A) {\footnotesize$(-1,-1)$};
			\node[anchor=south west] at ([xshift=0.1cm,yshift=0.1cm]A) {\footnotesize$\overline{U}(I_1)^+$};
			\node at (A) {$\bullet$};
			\node[anchor=north east] at (F) {\footnotesize$(1,-2)$};
			\node at (F) {$\bullet$};
			\node[anchor=north west] at (D) {\footnotesize$(1,-1)$};
			\node[anchor=south west] at ([xshift=-0.1cm,yshift=3cm]D) {\footnotesize$\overline{U}(I_2)^+$};
			\node[anchor=south east] at ([xshift=-0.1cm,yshift=3cm]D) {\footnotesize$\overline{U}(I_3)^+$};
			\node at (D) {$\bullet$};
			\node at (H) {$\bullet$};
			\node[anchor=north east] at (H) {\footnotesize$(-1,4)$};
			\node[anchor=south east] at ([xshift=-0.3cm,yshift=0.4cm]H) {\footnotesize$\overline{U}(I_4)^+$};
		\end{tikzpicture}
	\end{center}
	\caption{The regions $\overline{U}(I_1)^+,\overline{U}(I_2)^+,\overline{U}(I_3)^+,\overline{U}(I_4)^+\subset M_\R\cong\R^2$ for the multi-fan of $\mathcal{F}^{(2)}_3$ with top-dimensional cones $\Sigma^{(2)}=\{I_1,I_2,I_3,I_4\}$. The positive regions in $\overline{DH}_{\D_{\mathcal{F}^{(2)}_3},\xi}$ are plotted in blue and the negative in red so that the purple regions are completely canceled.}
	\label{HirzSurfgDH}
\end{figure}
The regions $\overline{U}(I_j)^+$ are shown in Figure \ref{HirzSurfgDH}. 
We can see part of $\overline{U}(I_3)^+$ which is negative in the generalized Duistermaat-Heckman measure is not canceled out by any of the other positive regions. 
We are left with a non-convex multi-polytope for the support of $\overline{DH}_{\D_{\mathcal{F}^{(2)}_3},\xi}$, which is exactly what would result from the trans-polar construction 
\begin{align}
\text{supp}\overline{DH}_{\D_{\mathcal{F}^{(2)}_3},\xi}=(\Dp_{\mathcal{F}^{(2)}_3})^\nabla=\D_{\mathcal{F}^{(2)}_3}
\end{align}
The negative-area region $\D_\text{ext}=\D\setminus(\Dp)^\circ$ is called the \textit{extension} of the Newton multi-polytope \cite{BH16}.
 \end{example}

 \begin{example}[$\check{\mathcal{F}}^{(2)}_3$ multi-fan]
 	\label{MirrorHizmulti-fan}
	We now turn to the mirror of the previous example. That is, we define the one dimensional cones using the vertices of $\D_{\mathcal{F}^{(2)}_3}$ and show that we recover the non-convex spanning polytope. In other words, the roles of $\D$ and $\Dp$ are reversed. Our vertices are now
	\begin{align}
		v_1=(1,-2) \,\, , \,\, v_2=(-1,4) \,\, , \,\, v_3= (-1,-1) \,\, , \,\, v_4=(1,-1)
	\end{align} 
	 Define the top dimensional cones
	\begin{align}
		I_1=\{1,2\} \,\, , \,\, I_2=\{2,3\} \,\, , \,\, I_3=\{3,4\} \,\, , \,\, I_4=\{4,1\}
	\end{align}
	so that our multi-fan is given by $\Sigma^{(2)}=\{I_j\}_{j=1}^4$. Now, $I_4$ has an orientation that is reversed relative to the other cones due to the fact that the corresponding region in the polytope has negative area.
	 This can be taken into account by setting $w_1=w_2=w_3=1$ while $w_4=-1$. Then, the generalized Duistermaat-Heckman measure becomes
	\begin{align}
		\overline{DH}_{\Dp_{\mathcal{F}^{(2)}_3},\xi}=\mathbbm{1}_{\overline{U}(I_1)^+}+\mathbbm{1}_{\overline{U}(I_2)^+}-\mathbbm{1}_{\overline{U}(I_3)^+}+\mathbbm{1}_{\overline{U}(I_4)^+},
	\end{align}
	where we chose $\xi=(1,0)$. 
	\begin{figure}[h!]
		\begin{center}
			\begin{tikzpicture}[scale=2]
				\draw[->,blue] (-3, -1) coordinate(A)-- (-1, 0)coordinate(B);
				\draw[->,blue] (A) -- (1, 0)coordinate(C);
				\draw[->,red, dashed] (C) -- ([xshift=1.6cm,yshift=0.4cm]C);
				\draw[->,red, dashed] (C) -- ([xshift=1cm,yshift=-1cm]C);
				\draw[->,red,dashed] (0,1)coordinate(D) -- ([xshift=1cm,yshift=1cm]D);
				\draw[->,red,dashed] (D) -- ([xshift=1cm,yshift=-1cm]D);
				\draw[->,blue] (B)-- ([xshift=1cm,yshift=1cm]B);
				\draw[->,blue] (B) -- ([xshift=1.33cm,yshift=0.66cm]B);
				\draw[opacity=0.2,fill=red] (D) -- ([xshift=1cm,yshift=1cm]D) -- ([xshift=1.6cm,yshift=0.4cm]C) -- (C) -- cycle;
				\draw[opacity=0.2,fill=blue] (A) -- ([xshift=2.66cm,yshift=1.33cm]B)-- ([xshift=1.6cm,yshift=0.4cm]C) -- (C) --cycle;
				\draw[opacity=0.2,fill=red] (C) -- ([xshift=1.6cm,yshift=0.4cm]C) -- (3.66,-0.8)coordinate(E)-- ([xshift=1cm,yshift=-1cm]C) --cycle;
				\draw[opacity=0.2,fill=blue] (C) -- ([xshift=1.6cm,yshift=0.4cm]C) -- (E)-- ([xshift=1cm,yshift=-1cm]C) --cycle;
				\draw[opacity=0.2,fill=blue] (B)-- ([xshift=2cm,yshift=2cm]B) -- ([xshift=2.66cm,yshift=1.33cm]B) -- cycle;
				\node[anchor=north east] at (A) {\footnotesize$(-3,-1)$};
				\node[anchor=south west,rotate=20] at ([xshift=1.33cm,yshift=0.25cm]A) {\footnotesize$\overline{U}(I_1)^+$};
				\node at (A) {$\bullet$};
				\node[anchor=south east] at (B) {\footnotesize$(-1,0)$};
				\node[rotate=30] at ([xshift=0.9cm,yshift=0.66cm]B) {\footnotesize$\overline{U}(I_4)^+$};
				\node at (B) {$\bullet$};
				\node[anchor=north,xshift=-0.1cm,yshift=-0.2cm] at (C) {\footnotesize$(1,0)$};
				\node[rotate=-40,anchor=north west] at ([xshift=0.8cm,yshift=0.1cm]C) {\footnotesize$\overline{U}(I_2)^+$};
				\node at (C) {$\bullet$};
				\node[anchor=south east] at (D) {\footnotesize$(0,1)$};
				\node[rotate=-40] at ([xshift=2cm,yshift=0.66cm]B) {\footnotesize$\overline{U}(I_3)^+$};
				\node at (D) {$\bullet$};
			\end{tikzpicture}
		\end{center}
		\caption{The regions $\overline{U}(I_1)^+,\overline{U}(I_2)^+,\overline{U}(I_3)^+,\overline{U}(I_4)^+\subset M_\R\cong\R^2$ for the multi-fan of $\check{\mathcal{F}}^{(2)}_3$ with top-dimensional cones $\Sigma^{(2)}=\{I_1,I_2,I_3,I_4\}$. The positive regions in $\overline{DH}_{\Dp_{\mathcal{F}^{(2)}_3},\xi}$ are plotted in blue and the negative in red so that the purple regions are completely canceled.}
		\label{MirrorHirzSurfgDH}
	\end{figure}
The regions $\overline{U}(I_j)^+$ are shown in Figure \ref{MirrorHirzSurfgDH}. Again, we see that mirror duality of these non-convex polytopes is restored by the trans-polar operation.
\begin{align}
	\text{supp}\overline{DH}_{\Dp_{\mathcal{F}^{(2)}_3},\xi}=(\D_{\mathcal{F}^{(2)}_3})^\nabla=\Dp_{\mathcal{F}^{(2)}_3}
\end{align}
It is interesting to note that the VEX point $u_{I_4}=(-1,0)$ corresponds to the top-dimensional cone $I_4$ with negative weight.
\end{example}

In these two dimensional examples, we are able to compute the support of the generalized Duistermaat-Heckman measure by finding the intersection of the dual cones $\overline{U}(I)^+$.
In higher dimensions, or with more vertices, it would be inefficient to do this computation without more tools.
To this end, we introduce an algebraic structure in the next section and we present a method to compute the support of $\overline{DH}_{\D,\xi}$ using the tropical hyperplanes $\beta_{v_\rho}$ associated to the monomials $Y^{v_\rho}$. 
This method is shown to produce $\text{supp}\,\overline{DH}_{\D,\xi} = \D$ for the non-Fano cases in question.
With the graded ring $H^{\text{trop}}_*(\Sigma)$, we will be able to compute the support of the generalized Duistermaat-Heckman measure for multi-polytopes of arbitrary dimensions.
Later, the generalized Duistermaat-Heckman measure will be used to equate the euclidean volume of a multi-polytope to the symplectic volume of the associated non-Fano toric variety.
\subsection{Graded ring generated by tropical hyperplanes}
\label{sec:gDH2}
Tropical geometry is piece-wise linear algebraic geometry where we replace the base commutative ring with the tropical semiring $(\R^{\text{trop}},\oplus_0,\otimes_0)$. Maslov dequantization realizes the \textit{tropical limit} $z\rightarrow 0$ using logarithms \cite{Mik04}. Namely, it takes the operations 
\begin{samepage}
	\begin{align}
		x\oplus_z y &= -\log_z\left(z^{-x}+z^{-y}\right)\nopagebreak\\
		x\otimes_z y &= -\log_z\left(z^{-x}z^{-y}\right)
	\end{align}
\end{samepage}
to their tropical counterparts
\begin{align}
	x\oplus_0 y &= \min\{x,y\} \\
	x\otimes_0 y &= x+y.
\end{align}
There are no inverses for $\oplus_0$, and we include the additive identity $\infty$, which makes $(\R^{\text{trop}},\oplus_0,\otimes_0)$ a semiring. 
Further, the map $x\mapsto -x$ is a semiring isomorphism, so we could also define $x\oplus_0 y = \max \{x,y\}$, but we  use $\min$ in what follows. Algebraic geometry over $\R^\text{trop}$ consists of studying varieties of tropical polynomials, i.e. polynomials where $\oplus_0$ and $\otimes_0$ are the operations. If $f\in\R^\text{trop}[x_1,\dots,x_n]$ is a tropical polynomial written as a sum of monomials $\beta=\{\beta_1,\dots,\beta_p\}$, then  the variety of $f$ is
\begin{align}
	V(f)=\text{Sing}(\min\{0,\beta\}).
\end{align}
The singular locus of the minimum of the monomials $\text{Sing}(\min\{0,\beta\})$ is all the points in $\R^n$ where two or more $\beta_i$ are tied for the minimum in $\beta$, or where $0=\beta_i$ is the minimum of $\beta$. In the tropical limit, all the $\beta_i$ are linear, so $V(f)$ is a weighted polyhedral complex. The regions where $\min\{0,\beta\}=\beta_i=0$ are faces of the polyhedral complex $V(f)$ and the regions where $\min\{0,\beta\}=\beta_i=\beta_j$ are asymptotic directions. This is illustrated in Figure \ref{CP2TropV}.
\begin{figure}[h!]
	\begin{center}
			\begin{tikzpicture}[scale=1.2]
			\draw (-1,-1)coordinate(A) -- (-1,2)coordinate(B);
			\draw (A) -- (2,-1)coordinate(C);
			\draw (B) -- (C);
			\draw (-2,-2)coordinate(D) -- (A);
			\draw (-2,4)coordinate(E) -- (B);
			\draw (C) -- (3,-1.5)coordinate(F);
			\node at (A) {$\bullet$};
			\node[anchor=south east] at (A) {$(-1,-1)$};
			\node at (B) {$\bullet$};
			\node[anchor=east] at (B) {$(-1,2)$};
			\node at (C) {$\bullet$};
			\node[anchor=south west] at (C) {$(2,-1)$};
			\node[anchor=east] at ([yshift=1.5cm]A) {$\beta_1=0$};
			\node[anchor=north] at ([xshift=1.5cm]A) {$\beta_2=0$};
			\node[anchor=west] at ([xshift=1.5cm,yshift=1.6cm]A) {$\beta_3=0$};
		\node[anchor=east] at (D) {$\beta_1=\beta_2$};
		\node[anchor=east] at (E) {$\beta_1=\beta_3$};
		\node[anchor=north] at (F) {$\beta_2=\beta_3$};
		\end{tikzpicture}
	\end{center}
\caption{The tropical variety $V(f)$ of the tropical polynomial 
		$f(y_1,y_2)=0\oplus_0\beta_1 \oplus_0 \beta_2\oplus_0 \beta_3$
	Notice the faces of the polyhedron are given by $\beta_i=0$ and the asymptotic directions are given by $\beta_i=\beta_j$ where $\beta_1=1+y_1, \beta_2=1+y_2,$ and $\beta_3=1-y_1-y_2$.}
\label{CP2TropV}
\end{figure}
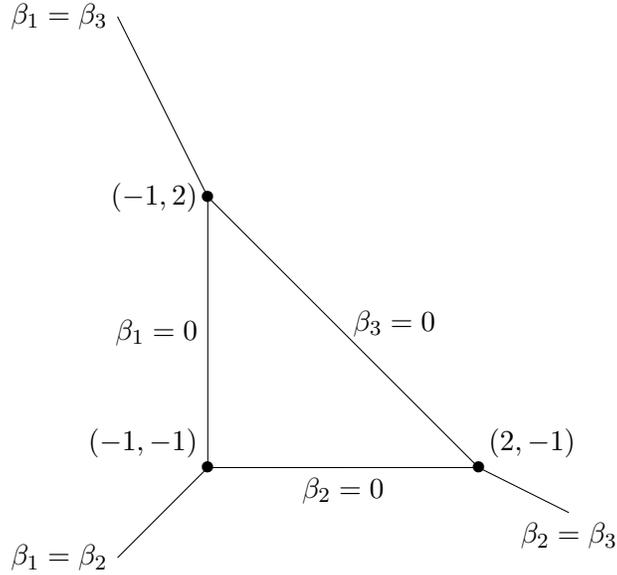
\par The utility of tropical geometry in string theory became apparent through the Gross-Siebert program of mirror symmetry \cite{GS}. Taking an SYZ mirror symmetry perspective \cite{SYZ} and using a Gross-Siebert tropical degeneration, Abouzaid-Ganatra-Iritani-Sheridan \cite{Sheridan_etal} decompose the real Lagrangian cycle $\check{C}^+_z=\check{X}_z\cap\R_+^n$ and compute its period $\pi(\check{C}^+_z)$. The $\widehat{\Gamma}$-conjecture for Calabi-Yau hypersurfaces in toric varieties, given by Equation \ref{GammaEqn}, is proven in the large radius/large complex structure limit by using this decomposition and integrating the Duistermaat-Heckman measure over the \textit{tropical amoeba} 
\begin{align}
	\mathcal{A}^\text{trop}(\check{C}^+_z)=\lim_{z\rightarrow 0} \Log_z (\check{C}^+_z)
\end{align}
of the real Lagrangian cycle, which is a limit of the amoebas $\mathcal{A}^z(\check{C}^+_z)$ defined similarly but without the limit.
Here,
\begin{samepage}
	\begin{align}
		\Log_z:(\C^*)^{n+1}&\longrightarrow\R^{n+1}\\
		(Y_1,\dots,Y_{n+1})&\longmapsto(y_1,\dots,y_{n+1})=(\log_z |Y_1|,\dots,\log_z |Y_{n+1}|) \nn
	\end{align}
\end{samepage}
 is a moment map on $(\C^*)^{n+1}$. It also is shown that the unique compact connected component of the amoeba converges to boundary of the Newton polytope of the toric ambient space $\mathcal{Y}_\D$
 \begin{align}
 	\mathcal{A}^\text{trop}(\check{C}^+_z) \sim \partial\D,
 \end{align}
which becomes the base of an SYZ fibration\footnote{The space $\mathcal{Y}_{\Dp}$ is not a toric variety since it was constructed from a non-trivial multi-fan structure. Rather, it is a torus manifold \cite{HattoriMasuda}.}
\begin{equation}
	\label{SYZdigaram}
	\begin{tikzcd}
		& \mathcal{Y}_\Delta & &\mathcal{Y}_{\Dp} & \\
		\mathscr{O}_X \arrow[r, two heads] & X \arrow[u,hook] \arrow[dr, swap, "p"] & & \check{X}_z \arrow[dl, "\check{p}"] \arrow[u,hook] & \check{C}^+_z \arrow[l, hook'] \\
		& & B\cong\partial\Delta & & 
	\end{tikzcd}
\end{equation}

\begin{figure}[h!]
	\begin{center}
		\includegraphics[scale=0.4]{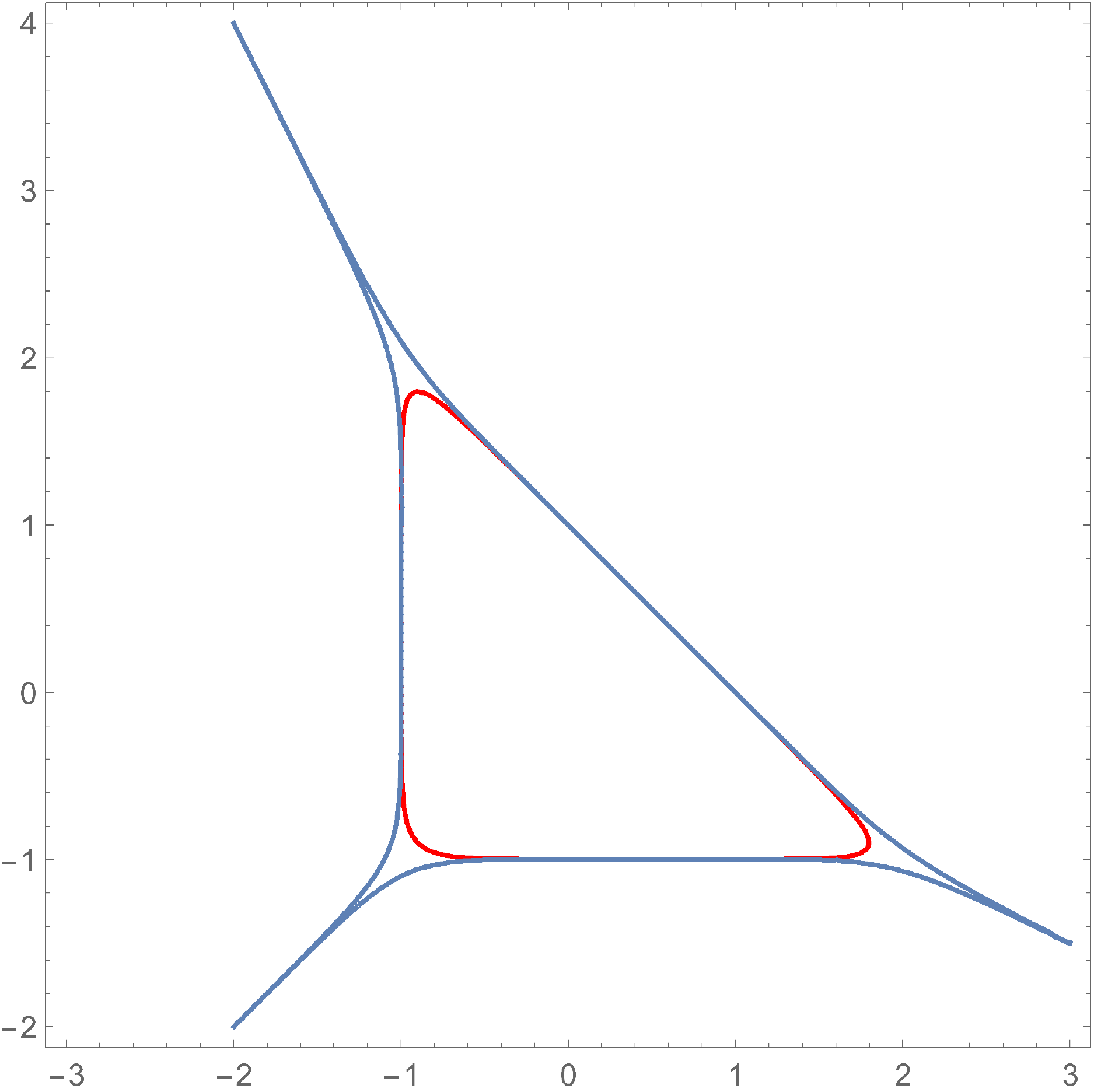}
	\end{center}
\caption{The amoeba of the real Lagrangian cycle $\mathcal{A}^z(\check{C}^+_z)$ for $\CP^2$ with $W_z=1-z(Y_1+Y_2+(Y_1Y_2)^{-1})$ is plotted in red for $z=0.001$. In blue we have the (boundary of the) amoeba of the entire mirror manifold $\mathcal{A}^z(\check{X})$.}
\label{CP2amoeba}
\end{figure}
We extend these results to non-Fano toric ambient spaces defined by multi-polytopes, such as $\mathcal{Y}_\D=\mathcal{F}^{(3)}_3$, using the generalized Duistermaat-Heckman measure discussed in the previous section. To this end, we tropicalize the superpotential and describe how it can be used to construct the generalized Duistermaat-Heckman measure using tropical hyperplanes.
\par We can write the tropicalized superpotential\footnote{If there are multiple complex structure moduli $(z_1,\dots,z_s)$, we rewrite $W_z(Y)$ in a common base $z_i=z$. Then $y_i=\log_{z}|Y_i|$ and $z_j=z^{\log z_j}$.}
\begin{align}
	W^\text{trop}_z(y)=\sum_{\rho\in\Sigma^{(1)}}z^{\beta_\rho(y)}
\end{align}
where $y=(y_1,\dots,y_{n+1})\in\R^{n+1}$ and
\begin{align}
	\beta_\rho(y)=\lambda_\rho+v_\rho\cdot y
\end{align}
is the tropical monomial corresponding to $v_\rho$. It deserves such a name since it is a linearization of the monomial $z^{\lambda_\rho}Y^{v_\rho}$ and 
\begin{align}
	\log_z\big|z^{\lambda_\rho}Y^{v_\rho} \big|=\beta_\rho(\Log_z(Y)).
\end{align}
The coefficients $\lambda_\rho$ are determined by the K\"{a}hler form
\begin{align}
	\omega=\sum_{\rho\in\Sigma^{(1)}}\lambda_\rho D_\rho,
\end{align}
and since we are looking at the anticanonical divisor $K^*$, we set every $\lambda_\rho=1$. We drop the coordinate $y$ to denote the region ``above" the corresponding hyperplane
\begin{align}
	\beta_\rho=\{\lambda_\rho+v_\rho\cdot x\ge0\}\subset\R^{n+1}.
\end{align}
Multiple constraints are denoted by multiple indices.
\begin{align}
	\beta_{\rho_1\rho_2}=\{\lambda_{\rho_1}+v_{\rho_1}\cdot x\ge0 \,\, , \,\, \lambda_{\rho_2}+v_{\rho_2}\cdot x\ge0\}=\beta_{\rho_1}\cap\beta_{\rho_2}\subset\R^{n+1}.
\end{align}
An index with a tilde corresponds to the region ``below" the corresponding hyperplane
\begin{align}
	\beta_{\tilde{\rho}}=\{\lambda_\rho+v_\rho\cdot x\le0\}=\mathbbm{1}-\beta_\rho\subset\R^{n+1}.
\end{align}
Here, $\mathbbm{1}\cong \R^{n+1}$ is the identity with respect to intersection, since 
\begin{align}
	\beta_\rho\cap\mathbbm{1}=\beta_{\rho}.
\end{align}
We call $\beta_{i_1\dots i_k}$ a \textit{tropical polytope} with $k$ constraints. 
When $k=1$, $\beta_{i_1}$ is called a \textit{tropical hyperplane}.
When $k$ is the dimension of the toric variety, we call $\beta_{i_1\dots i_{n+1}}$ a \textit{tropical cone}.
The set of tropical polytopes with $k$ constraints  be denoted by $H^\text{trop}_k(\Sigma)$ to emphasize the fact that their intersection data is encoded in the multi-fan $\Sigma$. 
We  define a graded ring structure on $H^\text{trop}_*(\Sigma) = \oplus_k H^\text{trop}_k(\Sigma)$ which  allow us to compute the support of $\overline{DH}_{\D,\xi}$. This is because the measure is written as a weighted sum of the characteristic functions $\mathbbm{1}_{\overline{U}(I)^+}$ and
\begin{align*}
	\text{supp}\,\mathbbm{1}_{\overline{U}(I)^+}=\beta_{i_1i_2\dots i_{n+1}}
\end{align*}
where $i_j\in I$ and the $i_j$  get a tilde if $i_j\notin E_I$. 
We first consider the simple multi-polytope (i.e. polytope) case before moving to the multi-polytopes which define non-Fano toric varieties.
\begin{proposition}
	If $\Delta$ is reflexive and $\Delta^*=\emph{Conv}\{v_1,\dots,v_k\}$, then
	\begin{align}
		\Delta=\beta_{12\dots k}
	\end{align}
\end{proposition}
\begin{proof}
	This is essentially proven in \cite[Proposition 4.1]{Nis06}, but here we are using new notation. We take $\Sigma$ to be the normal fan so that $\Sigma^{(1)}=\{v_1,\dots,v_k\}$. 
	By construction, $v_i$ is normal to the $i$th face of $\Delta$, and so the $i$th face is given by
	\begin{align}
		F_i=\{1+v_i\cdot x = 0\}
	\end{align}
	To fill out the interior of $\Delta$, we take the common positive span of the normals to these hyperplanes:
	\begin{align}
		\Delta
		&=\{1+v_1\cdot x \ge 0,\dots, 1+v_k\cdot x\ge0 \}=\beta_{12\dots k} \nonumber
	\end{align}
\end{proof}
To compute sums, differences and intersections of the $\beta_\rho$'s we need to derive the algebraic rules that they satisfy. 
To this end, the following propositions are proven. 
A region with a minus sign is interpreted as having reversed orientation, or negative volume. The propositions below generalize to multi-indices, but we  only need to make use of what was proved.
\begin{proposition}
	The difference between two regions which share $k$ constraints is
	\begin{align}
		\beta_{i_1\dots i_k}-\beta_{i_1\dots i_k i_{k+1}}=\beta_{i_1\dots i_k \tilde{i}_{k+1}}
	\end{align}
\end{proposition}
\begin{proof}
	Since these regions share the hyperplane constraints defined by normal vectors $v_{i_1},\dots,v_{i_k}$, 
	we have that $\beta_{i_1\dots i_k i_{k+1}}\subset\beta_{i_1\dots i_k}$.
	The region difference can then be interpreted as the set difference, 
	which can be computed in terms of the intersection of the compliment. 
	We write the compliment as a union of regions where one constraint is flipped, 
	and then the only region that survives the intersection is $\beta_{i_1\dots i_k \tilde{i}_{k+1}}$.
	\begin{align*}
		\beta_{i_1\dots i_k}-\beta_{i_1\dots i_k i_{k+1}}&=	\beta_{i_1\dots i_k}\cap\left(\beta_{i_1\dots i_k i_{k+1}}\right)^c\\
		&=\beta_{i_1\dots i_k}\cap\left(\beta_{\,\tilde{i}_1\dots i_k i_{k+1}}\cup\dots\cup\beta_{i_1\dots \tilde{i}_k i_{k+1}}\cup\beta_{i_1\dots i_k \tilde{i}_{k+1}}\right)\\
		&=\beta_{i_1\dots i_k \tilde{i}_{k+1}}
	\end{align*}
\end{proof}

\begin{proposition}
	The difference between two regions which share one constraint is
	\begin{align}
		\beta_{ij}-\beta_{ik}=\beta_{ij\,\tilde{k}}-\beta_{i\,\tilde{j}\,k}
	\end{align}
\end{proposition}
\begin{proof}
	The proof is very similar, but now $\beta_{ik}\not\subset\beta_{ij}$. 
	Instead of a set difference, we now have a signed symmetric difference
    \begin{samepage}
	\begin{align*}
			\beta_{ij}-\beta_{ik}&=	\beta_{ij}\cap\left(\beta_{ik}\right)^c-\left(\beta_{ij}\right)^c\cap\beta_{ik}\\
		&=\beta_{ij}\cap\left(\beta_{\,\tilde{i}\,k}\cup\beta_{i\,\tilde{k}} \right)-\left(\beta_{\,\tilde{i}\,j}\cup\beta_{i\,\tilde{j}} \right)\cap\beta_{ik}\\
		&=\beta_{ij\,\tilde{k}}-\beta_{i\,\tilde{j}\,k}
	\end{align*}
    \end{samepage}
\end{proof}
\noindent We now present 5 calculations which utilize the algebraic relations proved above.
\begin{example}[gDH measure for $\CP^2$ using $H^{\text{trop}}_*(\Sigma)$]
	The calculation of $\text{supp}\overline{DH}_{\D_{\CP^2},\xi}$  as in Example \ref{P2miltifan} can now be done completely algebraically. Using the data from that example, we have
	\begin{align*}
		\text{supp}\overline{DH}_{\D_{\CP^2},\xi}&=\sum_{k=1}^3(-1)^{|E_{I_k}|}w_{I_k}\,\text{supp}\,\mathbbm{1}_{\overline{U}(I_k)^+}\\
		&= \beta_{12}+\beta_{\tilde{2}\,\tilde{3}}-\beta_{1\,\tilde{3}}\\
		&= \beta_{12}-\beta_{12\,\tilde{3}}\\
		&=\beta_{123}
	\end{align*}
The term $\beta_{\,\tilde{1}\,\tilde{2}\,\tilde{3}}=\emptyset$, so it was not included in the third line. We have now shown algebraically that $\text{supp}\overline{DH}_{\D_{\CP^2},\xi}=\D_{\CP^2}$.
\end{example}
\begin{example}[gDH measure for $\mathcal{F}^{(2)}_3$ using $H^{\text{trop}}_*(\Sigma)$]
	\label{HirzSurfBetas}
	Consider the non-Fano Hirzebruch surface $\mathcal{F}^{(2)}_3$, where
	\begin{align*}
		\Sigma^{(1)}=\{(-1,0),(1,0),(0,1),(-3,-1)\}=\{v_i\}_{i=1}^4.
	\end{align*}
	We can do a very similar calculation. Now our top-dimensional cones are 
	\begin{align}
		I_1&=\{1,3\}  & I_2&=\{1,4\}\\
		I_3&=\{2,3\}  & I_4&= \{2,4 \}\nonumber 
	\end{align}
	and our tropical cones for $\xi=(1,0)$ are
	\begin{align}
		\text{supp}\,\mathbbm{1}_{\overline{U}(I_1)^+}&=\beta_{\,\tilde{1}3}  & \text{supp}\,\mathbbm{1}_{\overline{U}(I_2)^+}&=\beta_{1\,\tilde{4}}\\
		\text{supp}\,\mathbbm{1}_{\overline{U}(I_3)^+}&=\beta_{23}  & \text{supp}\,\mathbbm{1}_{\overline{U}(I_4)^+}&=\beta_{\,\tilde{2}\,\tilde{4}}\,. \nonumber
	\end{align}
	These regions can be summed with appropriate signs to compute the support of the generalized Duistermaat-Heckman measure.
	\begin{align*}
		\text{supp}\,\overline{DH}_{\D_{\mathcal{F}^{(2)}_3},\xi}&=\sum_{k=1}^4(-1)^{|E_{I_k}|}w_{I_k}\,\text{supp}\,\mathbbm{1}_{\overline{U}(I_k)^+}\\
		&=-\beta_{\,\tilde{1}3}-\beta_{1\,\tilde{4}}+\beta_{23}+\beta_{\,\tilde{2}\,\tilde{4}}\\
		&=(\beta_{23}-\beta_{\,\tilde{1}3})-(\beta_{1\,\tilde{4}}-\beta_{\,\tilde{2}\,\tilde{4}})\\
		&=\beta_{123}-\beta_{12\,\tilde{4}}\\
		&=\beta_{1234}-\beta_{12\,\tilde{3}\,\tilde{4}}
	\end{align*}
	We can see that there is a positive part from $(\Dp)^\circ$ and a negative part from $\D_\text{ext}$. 
\end{example}

\begin{example}[gDH measure for $\check{\mathcal{F}}^{(2)}_3$ using $H^{\text{trop}}_*(\Sigma)$]
	\label{MirrorHirzBeta}
	The mirror of the previous example can be handled in the same manner, but now $w_4=-1$ from the reversed orientation cone shown in Figure \ref{fig:F3Multifan}. From Example \ref{MirrorHizmulti-fan} we can show
\begin{align*}
	\text{supp}\,\overline{DH}_{\Dp_{\mathcal{F}^{(2)}_3},\xi}=\beta_{234}+\beta_{1\,\tilde{2}\,3}
\end{align*}
It is necessary to express our result as a sum of two regions due to the non-convexity of $\Dp$. Here, both regions are positive unlike what we saw in the previous example for $\D$. Now that we have constructed $\D$ and $\Dp$ we can compute the volume of each. We find $\text{Vol}(\D)+\text{Vol}(\Dp)=12$ in unit triangles due to the famous 12 theorem. The 12 theorem is also equivalent to the statement that the second Todd class integrates to 1:
\begin{align}
	\int_{\mathcal{Y}_\D}\text{Td}_2(\mathcal{Y}_\D)=\frac{1}{12}\int_{\mathcal{Y}_\D} c_2+c_1^2=1
\end{align}
More precisely, the multi-polytope $\D$ for $\mathcal{F}^{(2)}_3$ from the previous example has $\text{Vol}(\D)=9-1=8$, with the negative area unit triangle coming from $\D_\text{ext}$. The spanning polytope has $\text{Vol}(\Dp)=4$, so the 12 theorem holds.
\end{example}

\begin{example}[gDH measure for $\mathcal{F}^{(3)}_3$ using $H^{\text{trop}}_*(\Sigma)$]
	\label{3dF3gDH}
	Moving on to three dimensions, we can compute the support of the generalized Duistermaat-Heckman measure for $\mathcal{F}^{(3)}_3$. 
	Using the data in Example \ref{ex:Hirzebruch3d}, we can construct the top dimensional cones
	\begin{align}
		I_1&=\{1,2,4\}  & I_2&=\{1,2,5\}\nonumber\\
		I_3&=\{1,3,4\}  & I_4&= \{1,3,5 \} \\
		 I_5&= \{2,3,4\} &  I_6&= \{2,3,5\}\nonumber
	\end{align}
These 6 cones were chosen from the $5 \choose {3}$ $ =10$ possible combinations so that $\Sigma^{(3)}=\{I_k\}_{k=1}^6$ defines a star triangulation of $\Delta^*$. 
Namely, the origin $(0,0,0)\in\R^3$ is the common point (star origin) of all the cones in $\Sigma^{(3)}$. 
Further, the generators of each cone are linearly independent and the vertices $u_I=\cap_{i\in I} F_i$ of the dual polytope $\D$ are all integral. 
There are two cones $\{2,4,5\}$ and $\{3,4,5\}$ where the corresponding intersection point $u_I$ would be $(\frac{5}{3},-1,-1)$ or $(-1,\frac{5}{3},-1)$, respectively, so we do not include these cones.
By computing the $E_I$ for $\xi=(1,2,3)$, we get our tropical cones
\begin{align}
	\text{supp}\,\mathbbm{1}_{\overline{U}(I_1)^+}&=\beta_{12\tilde{4}}  & \text{supp}\,\mathbbm{1}_{\overline{U}(I_2)^+}&=\beta_{\,\tilde{1}25}\nonumber\\
	\text{supp}\,\mathbbm{1}_{\overline{U}(I_3)^+}&=\beta_{1\,\tilde{3}\,\tilde{4}}  & \text{supp}\,\mathbbm{1}_{\overline{U}(I_4)^+}&=\beta_{\,\tilde{1}\,\tilde{3}\,5} \\
	\text{supp}\,\mathbbm{1}_{\overline{U}(I_5)^+}&=\beta_{\,\tilde{2}\,\tilde{3}\,\tilde{4}} &  \text{supp}\,\mathbbm{1}_{\overline{U}(I_6)^+}&=\beta_{235}\nonumber
\end{align}
We can easily compute the support of  $\overline{DH}_{\D,\xi}$ now that we have this set up. 
Using the definition of the generalized Duistermaat-Heckman measure \ref{gDH}, we have
\begin{samepage}
\begin{align*}
	\text{supp}\,\overline{DH}_{\D_{\mathcal{F}^{(3)}_3},\xi}&=\sum_{k=1}^6(-1)^{|E_{I_k}|}w_{I_k}\,\text{supp}\,\mathbbm{1}_{\overline{U}(I_k)^+}\\
	&=-\beta_{12\tilde{4}}-\beta_{\,\tilde{1}25}+\beta_{1\,\tilde{3}\,\tilde{4}}+\beta_{\,\tilde{1}\,\tilde{3}\,5}-\beta_{\,\tilde{2}\,\tilde{3}\,\tilde{4}}+\beta_{235}\\
	&=-\beta_{12\tilde{4}}+(\beta_{\,\tilde{1}\,\tilde{3}\,5}-\beta_{\,\tilde{1}25})+(\beta_{1\,\tilde{3}\,\tilde{4}}-\beta_{\,\tilde{2}\,\tilde{3}\,\tilde{4}})+\beta_{235}\\
	&=-\beta_{12\tilde{4}}+(\beta_{235}-\beta_{\,\tilde{1}235})+\beta_{12\,\tilde{3}\,\tilde{4}}\\
	&=(\beta_{12\,\tilde{3}\,\tilde{4}}-\beta_{12\tilde{4}})+\beta_{1235}\\
	&=\beta_{1235}-\beta_{123\tilde{4}}\\
	&=\beta_{12345}-\beta_{123\,\tilde{4}\,\tilde{5}}
\end{align*}
\end{samepage}
We can again see that the positive part corresponds to the usual polar operation $(\Dp)^\circ$, and the negative part corresponds to the extension $\D_\text{ext}$. 
\end{example}
\begin{example}[gDH measure for $\mathcal{F}^{(4)}_3$ using $H^{\text{trop}}_*(\Sigma)$]
	In 4 dimensions, there are now 6 vertices of the spanning polytope $\Dp$ for $\mathcal{F}^{(4)}_m$,
	\begin{align*}
		v_1=(-1,-1,-1,0) \,\, , \,\, v_2=(1,0,0,0) \,\, , \,\, v_3=(0,1,0,0) \,\, , \,\, v_4=(0,0,1,0) \,\, , \,\, v_5=(0,0,0,1),
	\end{align*}
with the sixth vertex containing the twisting parameter $v_6=(-m,-m,-m,-1)$. Out of the $6 \choose {4}$ $ =15$ possible top-dimensional cones, we choose 8 in the same fashion as the previous example. As a straightforward generalization of the calculation for $\mathcal{F}^{(2)}_m$ and $\mathcal{F}^{(3)}_m$, one can obtain 
\begin{align}
	\text{supp}\,\overline{DH}_{\D_{\mathcal{F}^{(4)}_3},\xi}=\beta_{123456}-\beta_{1234\,\tilde{5}\,\tilde{6}}
\end{align}
\end{example}

	\subsection{Calculating periods using the generalized Duistermaat-Heckman measure}
	\label{sec:gDH3}
	The Duistermaat–Heckman theorem equates the euclidean volume of $\D$ to the symplectic volume of $\mathcal{Y}_\D$ \cite[Theorem 2.10]{DHbook}.
	By definition, the volume of the Newton multi-polytope is given by the integral of the generalized Duistermaat-Heckman measure over all of $\R^{n+1}$, and this is in turn equal to $\text{Vol}(\mathcal{Y}_\D)$,
	\begin{align}
		\text{Vol}(\D)=\int_{\R^{n+1}} \overline{DH}_{\D,\xi}=\text{Vol}(\mathcal{Y}_\D).
	\end{align}
	The generalized Duistermaat-Heckman measure (as opposed to the Duistermaat-Heckman measure), allows us to do this calculation even when $\D$ is a multi-polytope.
	With careful application of the Duistermaat-Heckman theorem, we can compute the period $\pi(\check{C}^+_z)=\text{Vol}(\RLC)$ of the real Lagrangian cycle in the large radius/large complex structure limit:
	\begin{align}
		\label{BigInt}
		\int_{\check{C}^+_z}\check{\Omega}_z= \sum_{\nu,J,L}(-1)^{|J\,\text{\textbackslash} \,L|}\prod_\ell (-\log z_\ell)\int_{[0,\varepsilon]^{|J|}}\int_{\mathcal{Y}_\Delta}e^{(1-a)K^*-\sum_j b_j D_j}D_\nu\prod_j(-\log z_j) D_j\,db_j
	\end{align}
	Equation \ref{BigInt} is derived from \cite[Lemma 3.2]{Sheridan_etal} with $\omega\in c_1(\mathcal{Y}_\D)=K^*$, but here we are allowing a multi-polytope $\D$ and have taken into account the possibility of having multiple complex structure moduli $z=(z_1,\dots,z_s)\in\mathscr{M}_\text{CS}(\check{X})$. 
	Note this is necessary for the cases $\mathcal{Y}_\D=\mathcal{F}^{(n+1)}_m$.
	Let $V\subset\Dp$ be the set of vertices.
	To compute the above sum, which is over all the cones in the multi-fan $\Sigma$ of $\D$, we choose $\nu\in V$ and then sum over all $L\subset J \subset V$ with $\nu \notin J$.
	We then compute the products over $j\in J$ and $\ell \in L$.
	The geometric interpretation of the parameters $b_j$ is a small deformation of the K\"{a}hler class defined by $\D$ that takes us from the tropical limit $z_a\rightarrow 0$, where $\mathcal{A}^\text{trop}(\check{C}^{+}_z)\sim\partial\D$, to the cycle $\check{C}^+_z$.
	Lastly, since we are doing the computation at a small nonzero $z$, we need to rescale by the appropriate factor of $(-\log z_i)$ in order to get the volume of $\check{C}^{+}_z$ rather than the volume of $\mathcal{A}^z(\check{C}^{+}_z)$.
	For $n=2$ and $3$, we  expand the exponential to order $n-1$ to get an expression for $\pi(\check{C}^+_z)$ in terms of the toric divisors $D_i$ of the toric ambient space $\mathcal{Y}_\D$.
	Many terms go to $0$ either from not producing $(n+1)$ intersections or from being suppressed by $\varepsilon$.
	In what follows, Greek indices  run over all divisor classes, and the index $i$  run over independent divisor classes.
	\begin{proposition}[$\pi(\check{C}^+_z)$ for $n=2$]
    \label{prop:K3period}
    For K3 hypersurfaces, we have
			\begin{align}
				\label{K3period}
			\int_{\check{C}^+_z}\check{\Omega}_z=\frac{1}{2!}K^*\left(\sum_{i=1}^s (-\log z_i)D_i\right)^2-\frac{1}{2!}\zeta(2)K^*\sum_{\nu\neq \rho}D_\nu D_\rho
		\end{align}
	\end{proposition}
\begin{proof}
	To first order, the exponential can be written
	\begin{align}
		e^{(1-a)K^*-\sum_j b_j D_j}=1+(1-a)K^*-\sum_j b_j D_j.
	\end{align}
The only terms that  survive the product over $j\in J$ are terms where $|J|=1$ or $|J|=2$ since these  give triple intersection numbers and $\text{dim}_\C \mathcal{Y}_\D=n+1=3$. Here $a$ is a local coordinate
\begin{align}
	a(b)=-\log_z(1+z^{b}).
\end{align}
This form of the coordinate $a$ comes from the tropical decomposition of $\check{C}^+_z$ by Abouzaid-Ganatra-Iritani-Sheridan. 
The right hand side of Equation \ref{BigInt} is 0 unless the corresponding faces intersect $\{\beta_\nu(y)=0\}\cap\{\beta_j(y)=0\}_{j\in J}\neq\emptyset$. On such an intersection, the term $z^{\beta_\nu}=z^{a}$ in the superpotential $W_z^\text{trop}$,  give the dominant contribution. 
The terms $\{z^{\beta_j}\}_{j\in J}=\{z^{b_j}\}_{j\in J}$  give the subdominant contributions. Ignoring all other terms and rearranging $W_z^\text{trop}$, we obtain the local coordinate $a$ on $\mathcal{A}^z(\check{C}^+_z)$, since $W_z^\text{trop}$ is the defining equation of the amoeba.
We  need to integrate $a(b)$ over small perturbations $b\in[0,\varepsilon]$, so the following calculation  be useful
\begin{align}
	\label{aInt}
\lim_{\varepsilon\rightarrow0}\int_0^\varepsilon a(b)\,db=-\frac{1}{(-\log z)^2}\frac{\zeta(2)}{2}
\end{align}

	Define an open cover $\{E_\nu,J(b)\}$ of $\mathcal{A}^z(\check{C}^+_z)$ projected to the $(a,b)$ plane by
\begin{align}
	E_{\nu,J} (b)=\left\{y\in\R^{n+1} \,\bigg| \, z^{\beta_\nu(y)}+\sum_{j\in J} z^{\beta_j(y)}=1 \right\}, 
\end{align}
so that by the Duistermaat-Heckman theorem, the volume of the open sets with $|J|=1$ is
\begin{align}
	\text{Vol}(E_{\nu,j}(b))=\int_X (-a)D_\nu D_j =\int_{\mathcal{Y}_\D}\log_z(1+z^b)K^*D_\nu D_j
\end{align}
Then using Equation \ref{aInt}, the $\zeta(2)$ term in $\pi(\RLC)$ is
\begin{align}
\int_{\RLC|_\text{codim-2}}\check{\Omega}_z&=\sum_{\nu\neq j}(-\log z_j)^2\int_0^\varepsilon \text{Vol}(E_{\nu, j}(b))\, db\\
&=-\frac{1}{2!}\zeta(2)\sum_{\nu\neq\rho}K^*D_\nu D_\rho. \nn
\end{align}
The appropriate $(-\log z_j)^2$ gets canceled by $\int a(b)\, db$ which proves the $\zeta(2)$ term in the Proposition.
Such a contribution comes from the codimension-2 stratum (i.e. singular points) of the base $B$ of the SYZ fibration. We  discuss this more in the next section.
 The codimension-1 contributions vanish due to $c_1(X)=0$.
 The codimension-0 contributions that  survive the integral over $b$ are the ones proportional to $\varepsilon^2(\log z)^2$.
Terms of this form can have $|J|=2$ or $|J|=1$, and they are enumerated below. 
\begin{align}
	\int_{\RLC|_\text{codim-0}}\check{\Omega}_z&=\varepsilon^2\left(\sum_{\nu\neq j_1\neq j_2}(-\log z_{j_1})(-\log z_{j_2})D_\nu D_{j_1}D_{j_2}-\frac{1}{2}\sum_{\nu\neq j} (-\log z_j)^2D_\nu D_{j}^2\right) \nn \\
	&=\frac{1}{2!}K^*\left(\sum_{i=1}^s (-\log z_i)D_i\right)^2
\end{align}
We have rewritten $K^*=\sum_\nu D_\nu$ and reindexed the sum to give the volume term in the Proposition.
\end{proof}

	\begin{proposition}[$\pi(\check{C}^+_z)$ for $n=3$] 
    \label{prop:CY3Period}
    For Calabi-Yau 3-fold hypersurfaces, we have
	\begin{align}
		\label{CY3Period}
		\int_{\check{C}^+_z}\check{\Omega}_z=-\zeta(3)&\left(\frac{1}{3!}K^*\sum_{\nu\neq\rho_1\neq\rho_2}D_\nu D_{\rho_1} D_{\rho_2}-\frac{1}{2!}(K^*)^2\sum_{\nu\neq\rho}D_\nu D_\rho\right)\\
		&-\frac{1}{2!}\zeta(2)(K^*)^2\sum_{\nu\neq\rho}(-\log z_{i(\rho)})D_\nu D_\rho+\frac{1}{3!}K^*\left(\sum_{i=1}^s (-\log z_i)D_i\right)^3\nonumber.
	\end{align}
\end{proposition}

\begin{proof}
	We proceed in a similar fashion. To second order, the exponential can be written
	\begin{align}
		e^{(1-a)K^*-\sum_j b_j D_j}=1+(1-a)K^*-\sum_j b_j D_j+\left((1-a)K^*-\sum_j b_j D_j\right)^2
	\end{align}
so there now 5 different types of terms that can give us quadruple intersection numbers in Equation \ref{BigInt}, however terms of the form $-baD_\nu D_j^2$ are supressed by $\varepsilon$.
We still have the coordinate $a(b)$ in the first order part of the expansion, and now when $|J|=2$ we have a new local coordinate $a(b_1,b_2)$ of the same form
\begin{align}
	a(b_1,b_2)=\log_z(1+z^{-b_1}+z^{-b_2})
\end{align} 
Therefore, the following integrals appear as part of the expansion
\begin{align}
	\label{1zeta3}
	I^{\text{\rom{1}}}_{\zeta(3)}&=(-\log z)^3\lim_{\varepsilon\rightarrow 0} \int_0^\varepsilon a^2(b)\,db=\zeta(3) \\
	\label{2zeta3}
	I^{\text{\rom{2}}}_{\zeta(3)}&= (-\log z)^3 \lim_{\varepsilon\rightarrow 0} \int_0^\varepsilon \int_0^\varepsilon a(b_1,b_2)\,db_1\,db_2=-\zeta(3)
\end{align}
The contributions to $\pi(\RLC)$ that are proportional to $\zeta(3)$ in the expansion of the exponential are given by
\begin{align*}
	\int_{\RLC|_\text{codim-3}}\check{\Omega}_z&=I^{\text{\rom{2}}}_{\zeta(3)} K^*\sum_{\nu\neq j_1\neq j_2}D_\nu D_{j_1} D_{j_2}+\frac{1}{2}I^{\text{\rom{1}}}_{\zeta(3)}(K^*)^2\sum_{\nu\neq j}D_\nu D_j \\
	&=-\zeta(3)\left(\frac{1}{3!}K^*\sum_{\nu\neq\rho_1\neq\rho_2}D_\nu D_{\rho_1} D_{\rho_2}-\frac{1}{2!}(K^*)^2\sum_{\nu\neq\rho}D_\nu D_\rho\right)
\end{align*}
The contributions that are proportional to $\zeta(2)$ are given by
\begin{align}
	\int_{\RLC|_\text{codim-2}}\check{\Omega}_z&=-I_{\zeta(2)}(-\log z_{\ell_1})(-\log z_{\ell_2})(K^*)^2\sum_{\nu\neq j}D_\nu D_j(-\log z_j) \nn\\
	&=-\frac{1}{2!}\zeta(2)(K^*)^2\sum_{\nu\neq\rho}(-\log z_{i(\rho)})D_\nu D_\rho.
\end{align}
Here, the factors of $(-\log z_\ell)$ out front cancel the corresponding factors from $I_{\zeta(2)}$. Finally, the contributions that are proportional to $(-\log z_j)^3$ hit the constant term of the exponential expansion. They have $|J|=3$ and are of the form
\begin{align}
	\int_{\RLC|_\text{codim-0}}\check{\Omega}_z&=\sum_{\nu\neq j_1\neq j_2\neq j_3} D_\nu (-\log z_{j_1})(-\log z_{j_2})(-\log z_{j_3}) D_{j_1}D_{j_2}D_{j_3}\\
	&=\frac{1}{3!}K^*\left(\sum_{i=1}^s (-\log z_i)D_i\right)^3.\nn
\end{align}
\end{proof}

\begin{remark}
	By using the adjunction formula, one can express the right hand side of Equations \ref{K3period} and \ref{CY3Period} in terms of the Chern classes of the ambient space. To do this recall from Example \ref{ChernK3},
	\begin{align}
		c(X)=\frac{c(\mathcal{Y}_\D)}{c(NX)}\bigg|_X=\frac{1+c_1+\dots+c_{n+1}}{1+c_1}\bigg|_X
	\end{align}
	where the Chern classes $c_i$ refer to the ambient space. Since the total Chern class is given by the product $c(\mathcal{Y}_\D)=\prod_\rho (1+D_\rho)$, we can express $c_2(X)$ or $c_3(X)$ for K3s and Calabi-Yau 3-folds, respectively, in terms of the toric divisors:
	\begin{align}
		c_2(X)&=c_2-c_1^2= \sum_{\nu,\rho}D_\nu D_\rho-\left(\sum_\nu D_\nu\right)^2=\frac{1}{2!}\sum_{\nu\neq\rho}D_\nu D_\rho\\ 
		c_3(X)&=c_3-c_1c_2=\frac{1}{3!}\sum_{\nu\neq\rho_1\neq\rho_2}D_\nu D_{\rho_1} D_{\rho_2}-\frac{1}{2!}K^*\sum_{\nu\neq\rho}D_\nu D_\rho
	\end{align}
	Again pulling the computation back to the ambient space
	\begin{align}
		\chi=\int_X \, c_n(X)=\int_{\mathcal{Y}_\D} c_1(\mathcal{Y}_\D)c_n(X),
	\end{align}
	for K3s we have
	\begin{align}
		\int_{\check{C}^+_z}\check{\Omega}_z=\text{Vol}(X)-\zeta(2)\chi
	\end{align}
and for Calabi-Yau 3-folds we have
\begin{align}
	\int_{\check{C}^+_z}\check{\Omega}_z=\text{Vol}(X)-\zeta(2)\sum_{i=1}^s t_i\int_X \omega_i\wedge c_2(X)-\zeta(3)\chi
\end{align}
where we have used the mirror map $t_i=-\log z_i$ in the large radius/large complex structure limit.
Now one can easily see we have a direct match with the A-side of the $\widehat{\Gamma}$-conjecture, given by the left hand side of \ref{GammaEqn}.
\end{remark}

\begin{example}[$\pi(\check{C}^+_z)$ for $\mathcal{F}^{(3)}_m$]
	\label{HirzK3Period}
	Starting from Equation \ref{K3period}, we can utilize the Duistermaat-Heckman theorem to compute $\pi(\check{C}^+_z)$ for $\mathcal{F}^{(3)}_m$. We  see that our result matches with what we found using classical intersection theory in Example \ref{ChernK3}. The Duistermaat-Heckman theorem tells us that the intersection numbers
	\begin{align}
		K_{\nu\rho}=K^*D_\nu D_\rho=\begin{bmatrix}
			2(1-2m) & 2-m & 2-m & & 3 & &3 \\
			2-m & 2(1+m) & 2(1+m) & & 3 & & 3 \\
			2-m & 2(1+m) & 2(1+m) & & 3 & &3 \\
			3 & 3 & 3 & & 0 & & 0 \\
			3 & 3 & 3  & & 0 & & 0 \\
		\end{bmatrix}
	\end{align}
 are exactly the degree or ``length" of the edge $\D$ corresponding to $\beta_{\overline{\nu\rho}}=\{\beta_\nu(y)=0\}\cap\{\beta_\rho(y)=0\}$. When $m>2$ we can see that the degrees of the $\beta_{\overline{12}}$ and $\beta_{\overline{13}}$ edges become negative. The sum which computes the $\zeta(2)$ term of the period integral is
 \begin{align}
 	\widetilde{\text{Tr}}(K)=\frac{1}{2!}\sum_{\nu\neq\rho} K_{\nu\rho}
 \end{align}
In general we  use the tilde trace to denote the sum of all entries of a tensor where no two indices are equal, with a combinatorial factor out front.
\begin{align}
	\widetilde{\text{Tr}}(A_{\mu_1\dots\mu_r})=\frac{1}{r!}\sum_{\mu_1\neq\dots\neq\mu_r} A_{\mu_1\dots\mu_r}
\end{align} 
  For the $m$-twisted Hirzebruch 3-fold this becomes
\begin{align}
	\widetilde{\text{Tr}}(K) &=((2-m)+3)+(2(1+m)+(2-m)+3)+3(3)\nn\\
	&=2(2-m)+2(1+m)+6(3)\nn\\
	&=24
\end{align}
The Duistermaat-Heckman theorem has allowed us to use the multi-polytope $\D$ (which we constructed as $\text{supp}\overline{DH}_{\D,\xi}$ in Example \ref{3dF3gDH}) to compute a quantity associated to the anticanonical hypersurface $X\hookrightarrow\mathcal{Y}_\D$. 
Mirror symmetry of $X$
Similar to Example \ref{MirrorHirzBeta} where we had the 12 theorem, equivalent to $\int_{\mathcal{Y}_\D} \text{Td}_2(\mathcal{Y}_\D)=1$, we have the ``24 theorem"
\begin{align}
	\int_{\mathcal{Y}_\D} \text{Td}_3(\mathcal{Y}_\D)=\frac{1}{24}\int_{\mathcal{Y}_\D}c_1(\mathcal{Y}_\D)c_2(\mathcal{Y}_\D)=1.
\end{align}
For completeness, we also compute $\text{Vol}(X)$
\begin{align}
	\text{Vol}(X)=\frac{1}{2!}K^*\left(\sum_{i=1}^2 (-\log z_i)D_i\right)^2=(1+m)t_1^2+3t_1t_2
\end{align}
so that the period is
\begin{align}
	\int_{\RLC} \check{\Omega}_z = (1+m)t_1^2+3t_1 t_2-24\zeta(2)
\end{align}
\end{example} 

\begin{example}[$\pi(\check{C}^+_z)$ for $\mathcal{F}^{(4)}_3$]
	\label{HirzCY3Period} 
	In 4 dimensions, the degree of an edge  be given by a quadruple intersection number. Namely we  have two types of contributions 
\begin{align}
	K^{\text{\rom{1}}}_{\nu\rho\delta}&=K^*D_\nu D_\rho D_\delta\\
	K^{\text{\rom{2}}}_{\nu\rho}&=(K^*)^2D_\nu D_\rho
\end{align}
which we  call Type \rom{1} and Type \rom{2} respectively. Then from Equation \ref{CY3Period}, we can see the $\chi$ term can be obtained by subtracting the number of the Type \rom{2} contributions from the number of Type \rom{1} contributions.
\begin{align}
	\chi=\widetilde{\text{Tr}}(K^{\text{\rom{1}}})-\widetilde{\text{Tr}} (K^{\text{\rom{2}}})
\end{align}
For the $m$-twisted Hirzebruch 4-fold, we have
\begin{align}
	\label{TrK1}
	\widetilde{\text{Tr}}(K^{\text{\rom{1}}})&=(4D_1+(2-m)D_2)(4D_1^3+(12-3m)D_1^2D_2)=56
\end{align}
and
\begin{align}
	\label{TrK2}
	\widetilde{\text{Tr}}(K^{\text{\rom{2}}})&=(4D_1+(2-m)D_2)^2(6D_1^2+(8-3m)D_1D_2)=224
\end{align}
so that
\begin{align*}
	\chi=\widetilde{\text{Tr}}(K^{\text{\rom{1}}})-\widetilde{\text{Tr}} (K^{\text{\rom{2}}})=-168
\end{align*}
As we can see, there are 56 Type \rom{1} contributions and 224 Type \rom{2} contributions.
In terms of Chern classes, we have
\begin{align}
	\widetilde{\text{Tr}}(K^{\text{\rom{1}}})&=\int_{\mathcal{Y}_\D}c_1c_3\\
	\widetilde{\text{Tr}}(K^{\text{\rom{2}}})&=\int_{\mathcal{Y}_\D}c_1^2c_2
\end{align}
The $\zeta(2)$ term and the volume are also easily computed to give the period
\begin{align}
	\int_{\check{C}^+_z}\check{\Omega}_z=\frac{1}{3!}(2+3m)t_1^3+2t_1^2t_2-\zeta(2)((6m+44)t_1+24t_2)-\zeta(3)(-168)
\end{align}

\end{example}

	\section{Error in tropicalization}
	\label{sec:EIT}
	The $\chi$ term in $\pi(\check{C}^+_z)$ has an interpretation in terms of ``error in tropicalization". The period integral is done by approaching the large complex structure limit point $z=0$ of $\mathscr{M}_\text{CS}(\check{X})$ and then pulling back the computation to a small yet non-zero $z$. Therefore, there is a difference between integrating over $\overline{DH}_{\D,\xi}$ and the corresponding region of $\mathcal{A}^z({\check{C}^+_z})$.
	This difference is called error in tropicalization because in the true tropical limit $\mathcal{A}^\text{trop}({\check{C}^+_z})$ converges to $\partial\D$ and there would be no difference if we simply computed the volume of $\partial\D$.
	Following \cite{Sheridan_etal,Yamamoto22}, we  show how to get local $\zeta(n)$ contributions to $\pi(\check{C}^+_z)$ via error in tropicalization for $n=2$ and $n=3$.
	We  then comment on how new features of the error in tropicalization calculation arise when considering a non-Fano toric ambient space.
	With this alternate perspective, we can obtain the same Euler characteristic $\chi$ as we did using the generalized Duistermaat-Heckman measure.
	\subsection{Local contributions to the period} 
	Maslov dequantization takes us from logarithmic addition $x\oplus_z y = -\log_z(z^{-x}+z^{-y})$ to tropical addition $x\oplus_0 y = \min\{x,y\}$, so it is interesting to compute the difference between these two quantities. For simplicity, let's set $y=0$.
	\begin{align}
		\label{ErrorTrop}
		\int_\R \left(x\oplus_z y-x\oplus_0 y\right)\,dx&=-\int_\R\left(\log_z(z^{-x}+1)+\min\{x,0\}\right)\,dx\\
		&=-2\int_{\R_+}\log_z(z^{-x}+1)\,dx \nn\\
		&=-\frac{2}{(\log z)^2}\lim_{A\rightarrow \infty}\int^{z^{-A}}_1\frac{\log(1+b)}{b}\,db \nn\\
		&=-\zeta(2)/(\log z)^2 \nn
	\end{align}
	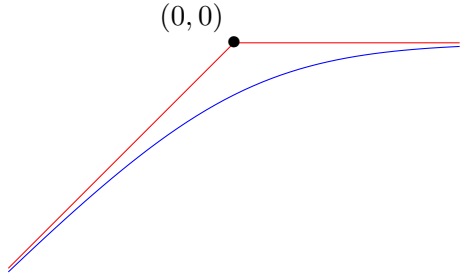
\begin{figure}[h!]
		\begin{center}
			\begin{tikzpicture}
				\draw[domain=-3:3, smooth, variable=\x, blue] plot ({\x}, {-ln(1+exp(-\x)});
				\draw[red] (3,0) -- (0,0);
				\draw[red] (-3,-3) -- (0,0);
				\node at (0,0) {$\bullet$};
				\node[anchor=south east] at (0,0) {$(0,0)$};
			\end{tikzpicture}
		\end{center}
	\caption{Logarithmic addition $x\oplus_z 0$ is plotted in blue and tropical addition $x\oplus_0 0$ is plotted in red for $x\in[-3,3]$ and $z=e$. For small $z$, the area bounded by these curves is approximately $\zeta(2)/(\log z)^2$}
	\end{figure}
	In the case of our period integral $\pi(\check{C}^+_z)$, we get a contribution of $-\zeta(2)$ for each such error in tropicalization integral from the factor of $(\log z)^2$ in Equation \ref{BigInt}.
	\par The contributions to $\pi(\check{C}^+_z)$ as in Equation \ref{ErrorTrop} appear at each singularity in the affine structure of the common base $B\cong\partial\D$ of the SYZ fibration, (see Diagram \ref{SYZdigaram}). 
	The affine structure can be constructed from rays $\R_+ v$ emanating from the origin and going through integral points $v$ of $\partial\D\sim\mathcal{A}^\text{trop}(\RLC)$.
	Along each edge of $\D$ connecting two adjacent integral points $v_i,v_j$, there is a singularity in the affine structure from an affine linear $\text{Aff}(\R^{n+1})=\R^{n+1}\rtimes\text{GL}_{n+1}(\R)$ change of coordinates from the coordinate system defined by $\R_+ v_i$ to the coordinate system defined by $\R_+ v_j$.
	The local form of the SYZ fibration $\check{p}:\check{X}_z\rightarrow B$ determines the type of contribution to $\pi(\RLC)$.
	This is because the local form of $\check{p}$ tells us how the amoebas $\mathcal{A}^z(\RLC)$ approach the tropical limit, and thus determines the form of the error in tropicalization integral.
	For a K3 surface, there are $\zeta(2)$ contributions of the form
	\begin{align}
		I_{\zeta(2)}=(-\log z)^2\int_0^\varepsilon \left(b\oplus_z 0-b\oplus_0 0\right)\,db=-\zeta(2)
	\end{align}
	For a Calabi-Yau 3-fold, there can be $I_{\zeta(2)}$ contributions as well as two different types of $\zeta(3)$ contributions,
	\begin{align}
		I^{\text{\rom{1}}}_{\zeta(3)}&=(-\log z)^3\int_0^\varepsilon  \left((b\oplus_z 0)^2-(b\oplus_0 0)^2\right)\,db=\zeta(3)\\
		I^{\text{\rom{2}}}_{\zeta(3)}&=(-\log z)^3 \int_0^\varepsilon \int_0^\varepsilon \left(b_1\oplus_z b_2 \oplus_z 0 - b_1\oplus_0 b_2\oplus_0 0 \right)\, db_1\, db_2=-\zeta(3)
	\end{align}
	In Table \ref{table:1}, we tabulate the local forms of $\check{p}$ that give rise to the above contributions.
	\begin{table}[h!]
		\centering
		\begin{tabular}{||c | c||} 
			\hline
			Local $\check{p}$ & $\pi(\RLC)$ contribution \\ [0.5ex] 
			\hline\hline
			$x_1+x_2=b\oplus_z 0$ & $I_{\zeta(2)}$ \\ 
			$x_1+x_2+x_3=b\oplus_z 0$ & $I^{\text{\rom{1}}}_{\zeta(3)}$ \\
			$x_1+x_2=b_1\oplus_z b_2\oplus_z 0$ & $I^{\text{\rom{2}}}_{\zeta(3)}$ \\
			\hline
		\end{tabular}
		\caption{Local contributions to $\pi(\RLC)$ via error in tropicalization.}
		\label{table:1}
	\end{table}
	Notice that these three error in tropicalization integrals can be obtained from the local integrals which appear in the proofs of the $\pi(\RLC)$ formulas (Equations \ref{aInt}, \ref{1zeta3}, and \ref{2zeta3} respectively). Further, the Type \rom{1} and Type \rom{2} contributions to the $\chi$ term of $\pi(\check{C}^+_z)$ for $n=3$ are exactly the counts of error in tropicalization integrals $I^{\text{\rom{1}}}_{\zeta(3)}$ and $I^{\text{\rom{2}}}_{\zeta(3)}$. 
	See \cite[Example 2.2]{Sheridan_etal} for an explicit computation of the error in tropicalization integrals for $\CP^3$. In what follows, we  show that $W_z^\text{trop}$ the non-Fano $\mathcal{F}^{(n+1)}_m$ can be cast into the local forms in Table \ref{table:1} along each edge of $\D$. There  be additional features that do not appear in the previously studied Fano cases. 
	\subsection{New contributions for non-Fano toric varieties}
	We have computed $\pi(\check{C}^+_z)$ for non-Fano toric varieties using the generalized Duistermaat-Heckman measure.
	Since this calculation used a tropical decomposition of $\mathcal{A}^z(\check{C}^+_z)$ in the large radius/large complex structure limit, one would expect to be able to obtain the $\chi$ term of the period via the error in tropicalization described in the previous section.
	 This does turn out to be the case, and there are novel features in the calculation.
	 In particular, there are error in tropicalization integrals that give contributions to the period with opposite signs, due to $\D_\text{ext}$ having flipped orientation.
	 We  compute such contributions for our examples $\mathcal{F}^{(3)}_3$ and $\mathcal{F}^{(4)}_3$, but first we go back to two dimensions for illustrative purposes.
	 \begin{remark}
	 	\label{CorrRemark}
	 	Mirror Landau-Ginzburg models for Hirzebruch surfaces are studied in \cite{CPS} using tropical geometry, which reproduced Auroux's previous results using an explicit deformation $\mathcal{F}^{(2)}_m\cong\mathcal{F}^{(2)}_{m+2k}$ \cite{Aur09}. Using the basis in Example \ref{ex:Hirzebruch3d} and the intersection numbers from Example \ref{GLSMintersect}, we can explicitly write Auroux's corrected superpotential $W^A$ for $\mathcal{F}^{(2)}_3$ from \cite[Proposition 2]{Aur09}.
	 	\begin{align}
	 		W^A_z(Y_1,Y_2)=\frac{z_1^2}{Y_1}+Y_1+Y_2+\frac{z_1^6}{z_2}\frac{1}{Y_1^3Y_2}+2\frac{z_1^4}{z_2}\frac{1}{Y_1^2}+\frac{z_1^2}{z_2}\frac{Y_2}{Y_1}
	 	\end{align}
	 	The tropicalized version can also be written down
	 	\begin{align}
	 		W^{A,\text{trop}}_z(y_1,y_2)=z_1^{2-y_1}+z_1^{y_1}+z_1^{y_2}+z_1^{6-3y_1-y_2-\log_{z_1}z_2}+2z_1^{4-2y_1-\log_{z_1}z_2}+z_1^{2-y_1+y_2-\log_{z_1}z_2}
	 	\end{align}
 	In what follows, we  use the na\"{i}ve superpotential
 	\begin{align}
 		W_z^\text{trop}=z^{1-y_1}+z^{1+y_1}+z^{1+y_2}+z^{1-3y_1-y_2}
 	\end{align}
 	and the appropriate generalization for  $\mathcal{F}^{(3)}_3$ and $\mathcal{F}^{(4)}_3$. This does not give the true mirror\footnote{For $\mathcal{F}^{(2)}_3$ it has been shown that there are more critical points of $W_z$ than Lagrangian submanifolds, so $m-2$ critical points need to be sent to infinity to match the dimensions of the corresponding moduli spaces \cite{FOOO}.} $\check{X}_z=W^{-1}_z(0)$, but we  show we can still obtain the correct $\chi$ term of $\pi(\RLC)$.
	 \end{remark}
 
	\begin{example}[Error in tropicalization for $\mathcal{F}^{(2)}_3$]
		We have already addressed the non-Fano Hirzebruch surface $\mathcal{F}^{(2)}_3$ several times, but now we take a slightly different approach. In Example \ref{HirzSurfBetas}, we found that
		\begin{align}
				\text{supp}\,\overline{DH}_{\D_{\mathcal{F}^{(2)}_3},\xi}=\beta_{1234}-\beta_{12\,\tilde{3}\,\tilde{4}}
		\end{align}
		In order to get a tropical variety with a region of opposite orientation, we propose the corresponding tropical polynomial should be of the form
		\begin{align}
			\label{HirzSurfTrop}
			f(y_1,y_2)=\left(0 \oplus \bigoplus_{i=1}^4 \beta_i(y) \right )\ominus(0\oplus \beta_1(y) \oplus \beta_2(y) \oplus \beta_{\,\tilde{3}}(y) \oplus \beta_{\,\tilde{4}}(y))
		\end{align}
	where $\beta_i(y)=1+v_i\cdot y$ and the $\beta_{\, \tilde{i}}(y)$ are formal inverses added to the semi-ring of tropical polynomials $\R^\text{trop}[y_1,\dots,y_{n+1}]$. A conjectured plot of the tropical variety is shown in Figure \ref{F3TropV}. 
	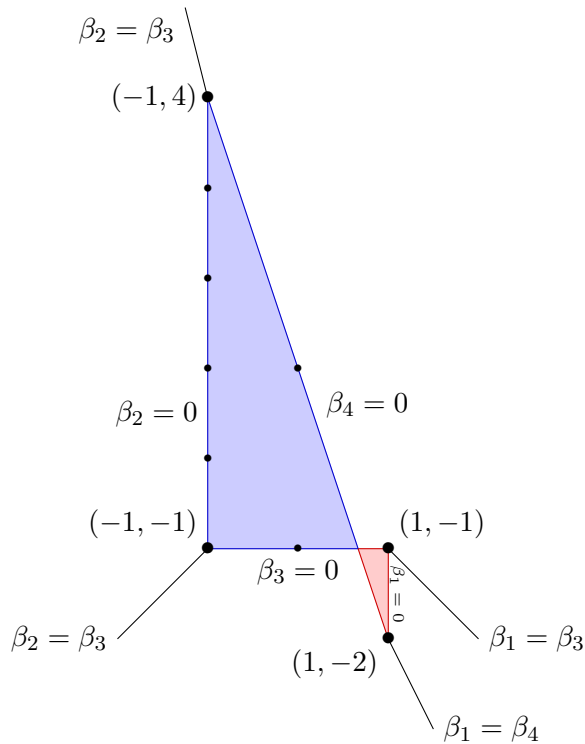
\begin{figure}[h!]
		\begin{center}
			\begin{tikzpicture}[scale=1.2]
				\draw[blue] (-1,-1)coordinate(A) -- (-1,4)coordinate(B);
				\draw[blue] (A) -- (0.666, -1)coordinate (N); 
				\draw[red] (N) -- (1,-1)coordinate(C);
				\draw[blue] (B) -- (N);
				\draw[red] (N) -- (1,-2)coordinate(D);
				\draw[red] (C) --(D);
				\draw[opacity=0.2,fill=red] (N) -- (C) -- (D) --cycle;
				\draw[opacity=0.2,fill=blue] (N) -- (B) -- (A) --cycle;
				\draw (-2,-2)coordinate(E) -- (A);
				\draw (-1.25,5)coordinate(F) -- (B);
				\draw (C) -- (2,-2)coordinate(G);
				\draw (D) -- (1.5,-3)coordinate(H);
				\node at (-1,0) {\tiny$\bullet$};
				\node at (-1,1) {\tiny$\bullet$};
				\node at (-1,2) {\tiny$\bullet$};
				\node at (-1,3) {\tiny$\bullet$};
				\node at (0,1) {\tiny$\bullet$};
				\node at (0,-1) {\tiny$\bullet$};
				\node at (A) {$\bullet$};
				\node[anchor=south east] at (A) {$(-1,-1)$};
				\node at (B) {$\bullet$};
				\node[anchor=east] at (B) {$(-1,4)$};
				\node at (C) {$\bullet$};
				\node[anchor=south west] at (C) {$(1,-1)$};
				\node at (D) {$\bullet$};
				\node[anchor=north east] at (D) {$(1,-2)$};
				\node[anchor=east] at ([yshift=1.5cm]A) {$\beta_2=0$};
				\node[anchor=north] at ([xshift=1cm]A) {$\beta_3=0$};
				\node[anchor=west] at ([xshift=1.2cm,yshift=1.6cm]A) {$\beta_4=0$};
				\node[rotate=270] at ([xshift=0.1cm, yshift=0.5cm]D) {\tiny$\beta_1=0$};
				\node[anchor=west] at (H) {$\beta_1=\beta_4$};
				\node[anchor=east] at (E) {$\beta_2=\beta_3$};
				\node[anchor=north east] at (F) {$\beta_2=\beta_3$};
				\node[anchor=west] at (G) {$\beta_1=\beta_3$};
			\end{tikzpicture}
		\end{center}
		\caption{The conjectured tropical variety $V(f)$ of the tropical polynomial defined in Equation \ref{HirzSurfTrop}. Notice the faces of the polyhedron are given by $\beta_i=0$ and the asymptotic directions are given by $\beta_i=\beta_j$ where $\beta_1=1-y_1, \beta_2=1+y_1, \beta_3=1+y_2$ and $\beta_4=1-3y_1-y_2$.}
		\label{F3TropV}
	\end{figure}
	Such a tropical variety would give rise to a multi-polytope $\D$ that we have shown satisfies the 12 theorem $\text{Vol}(\D)+\text{Vol}(\Dp)=12$. As we have seen, this euclidean volume is closely related to topological invariants of the corresponding toric ambient space $\mathcal{Y}_\D$. For elliptic curves $X$ the Euler characteristic is 0, so it does not make sense to compute $\pi(\RLC)$ as we have for K3s and Calabi-Yau 3-folds, but we can use this illustrative low-dimensional example to calculate the volume
	\begin{align}
		\text{Vol}(\D)=9(1)+1(-1)=8
	\end{align}
 	in unit triangles, which would be the leading order term in $\pi(\RLC)$. For this reason, we conjecture the above form of the tropical polynomial $f$.
		\begin{figure}[h!]
			\begin{center}
			\begin{tabular}{c c c c}
				\includegraphics[scale=0.2]{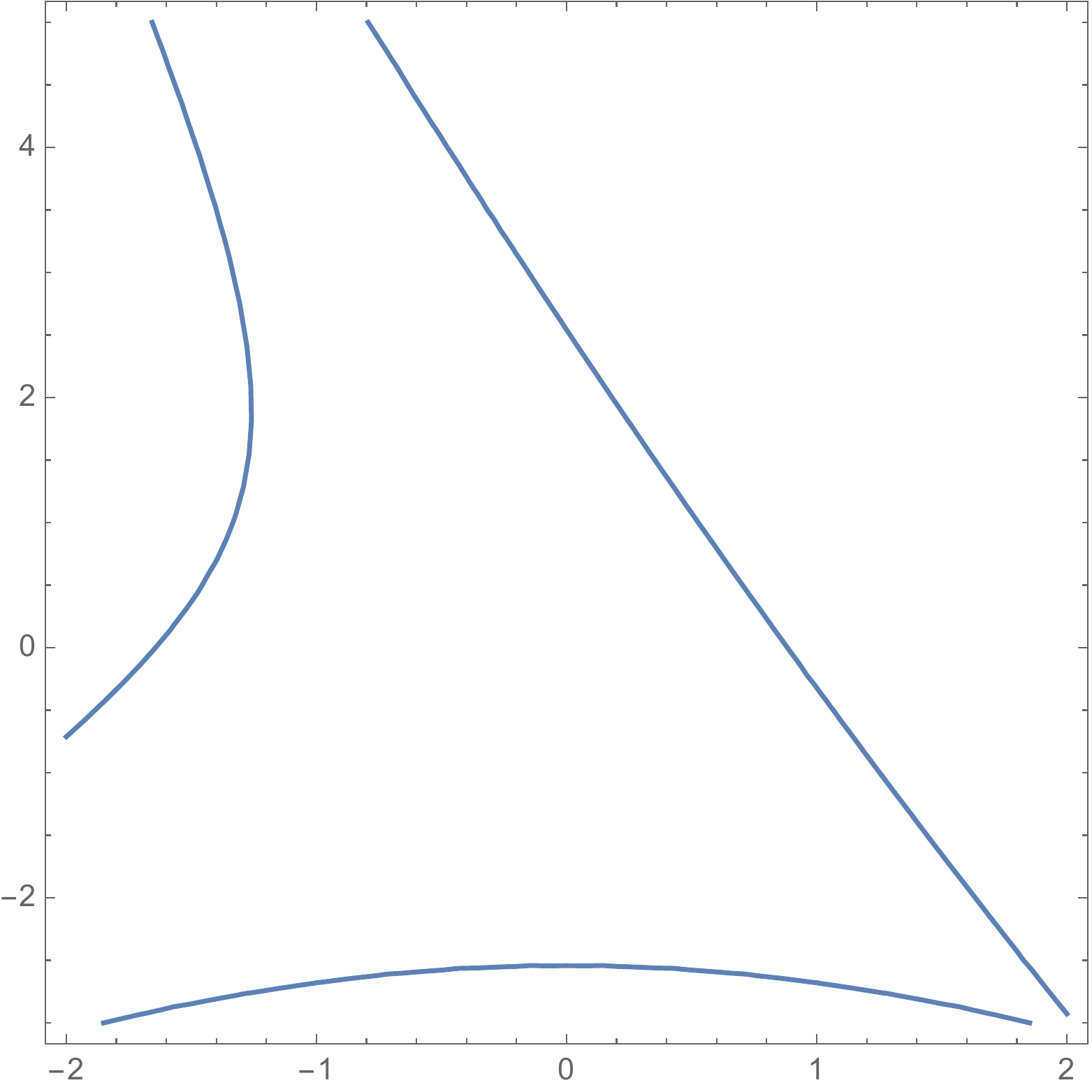}
				&
				\includegraphics[scale=0.2]{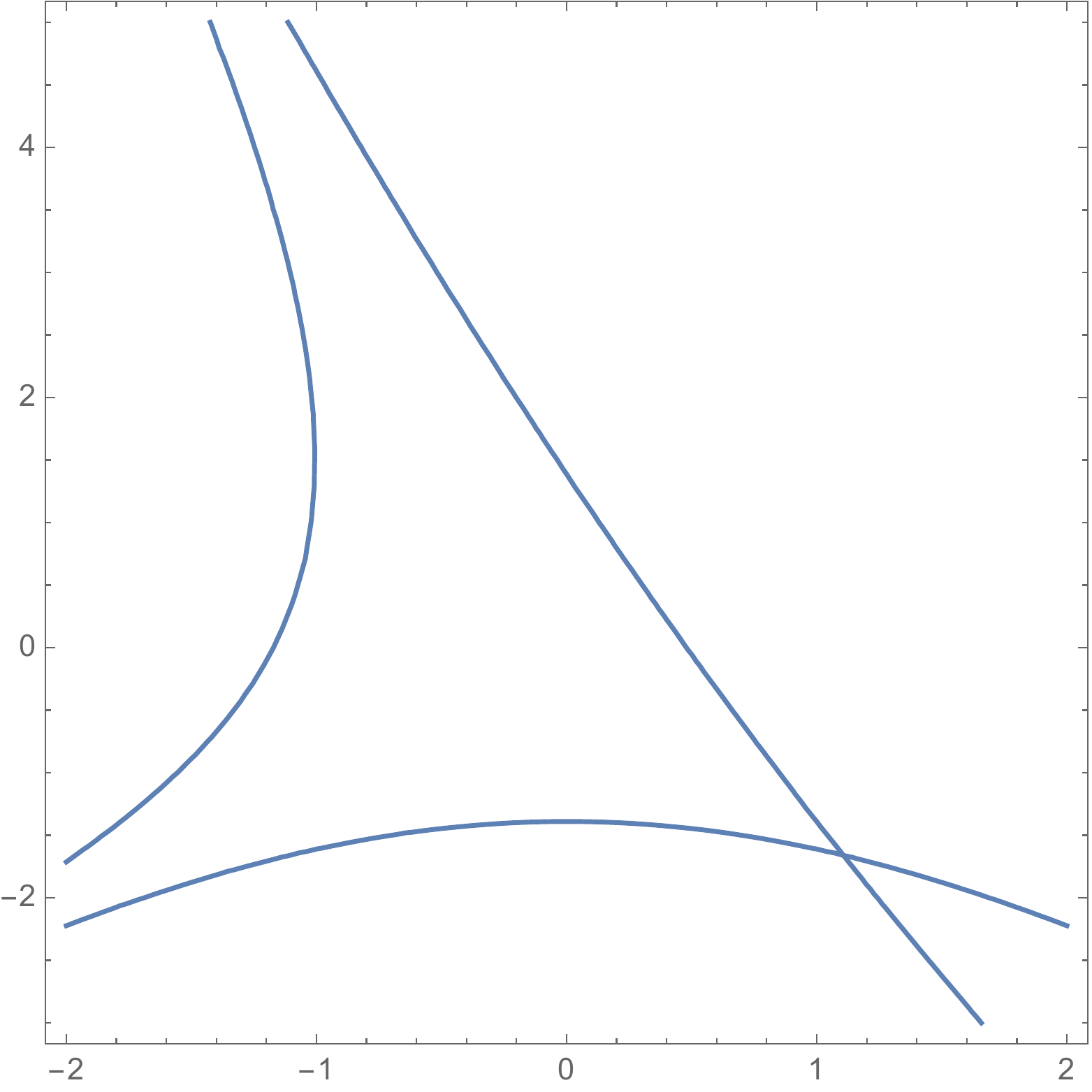}
				&
				\includegraphics[scale=0.2]{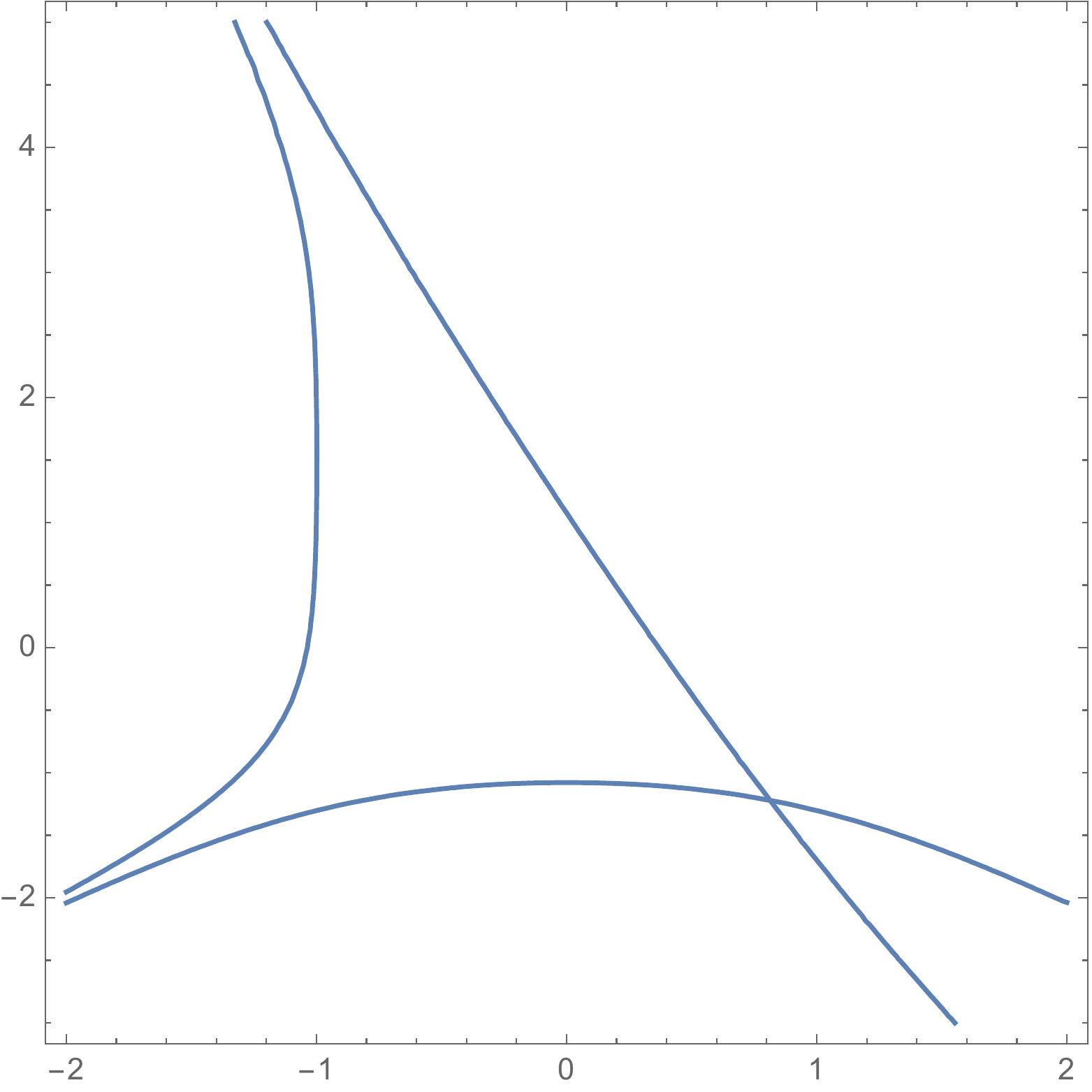}
				&
				\includegraphics[scale=0.2]{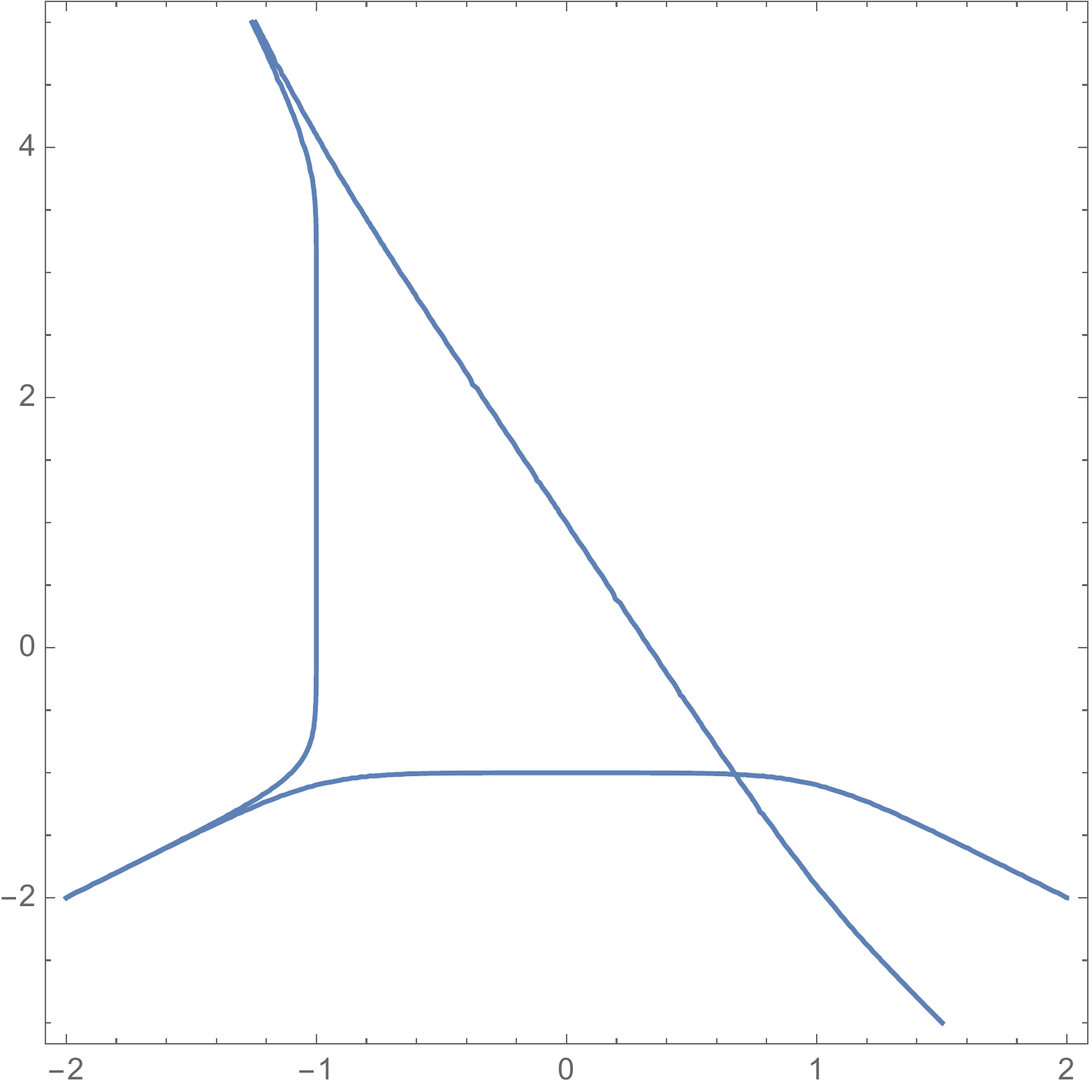}
			\end{tabular}
		\caption{The boundary of the amoeba $\mathcal{A}^z(\check{X}_z)$ for the mirror to $X\hookrightarrow\mathcal{F}^{(2)}_3$ at $z=0.6,0.3,0.1,$ and $0.001$. As we approch the tropical limit, we can see the amoeba ``twists" to give us $\D_\text{ext}$ as in Figure \ref{HirzSurfgDH}. The vertical boundary of $\D_\text{ext}$ is not present in this construction.}
		\label{fig:F3amoeba}
		\end{center}
		\end{figure}
	
	While we have showed that one can construct the multi-polytope $\D=(\Dp)^\nabla$ with the generalized Duistermaat-Heckman measure, obtaining $\mathcal{A}^\text{trop}(\RLC)\sim \partial\D$ via a tropical amoeba presents challenges. 
	In Figure \ref{fig:F3amoeba}, we can see that it is possible to manipulate the coefficients of $W_z^\text{trop}$ to obtain a ``twisted" amoeba for $\mathcal{F}^{(2)}_3$. 
	One can imagine the two dimensional amoeba as living in $\R^3$, so that the plot in the figure is non-planar. 
	Therefore the point $(2/3,-1)$ is not an intersection point of the $\beta_3$ edge with the $\beta_4$ edge, rather it is just an artifact of the projection to two dimensions.
	Without a more developed theory of the formal inverses $\beta_{\,\tilde{i}}$, it seems to be impossible to include the vertical edge of $\D_\text{ext}$ since this edge would need to intersect itself in the projection.
	Nonetheless, assuming a twisted amoeba of this nature can exist, we  now show that the $\chi$ term in $\pi(\RLC)$ can be obtained via error in tropicalization for $n=2$ and $n=3$.
	\end{example}
	
	\begin{example}[Error in tropicalization for $\mathcal{F}^{(3)}_3$]
		We have computed the $\chi$ term of $\pi(\RLC)$ for the Calabi-Yau hypersurface in $\mathcal{F}^{(3)}_3$ in Example \ref{HirzK3Period}, but here we  derive the same result in a different way. The tropicalized superpotential can be written as
		\begin{align}
			W_z^\text{trop}(y)=z^{1-y_1-y_2}+z^{1+y_1}+z^{1+y_2}+z^{1+y_3}+z^{1-3y_1-3y_2-y_3}
		\end{align}
	Along each edge of $\RLC=\{y\in\R^3 \, | \, W_z^\text{trop}(y)=1\}$, the defining equation can be cast into the form $x_1+x_2=b\oplus_z 0$.
	For example, along the segment between $(-1,-1,-1)$ and $(-1,-1,0)$ in the edge $\beta_{\overline{23}}=\{\beta_2(y)=0\}\cap\{\beta_3(y)=0\}$, we have $x_1=1+y_2+y_3$, $x_2=-y_3$, and $b=y_2-y_1$.
	If we were to count all the segments of $\Delta$, plotted in Figure \ref{3dF3Newton}, we would get 28. Since we get a contribution to the $\chi$ term of $\pi(\RLC)$ for each such segment, this would imply that the Euler characteristic of the K3 hypersurface would be 28. However, if we follow the edge $\beta_{\overline{34}}$ from $(-1,-1,-1)$ to $(2,-1,-1)$, we arrive at $\D_\text{ext}$. The vertical edge $\beta_{\overline{13}}$ from $(2,-1,-1)$ to $(2,-1,-2)$ has an orientation such that the error in tropicalization integral along that segment is given by $-I_{\zeta(2)}$. The local coordinates are
	\begin{align}
		x_1=1+y_2+y_3 \,\, , \,\, x_2=-y_3 \,\, , \,\, b=y_1+2y_2 
	\end{align}
	but we have to integrate in the opposite direction, giving the $\zeta(2)$ contribution with the opposite sign. Taking this into account, we can calculate the error in tropicalization
	\begin{align}
		(-\log z)^2\int_\varepsilon^0 \left(b\oplus_z 0-b\oplus_0 0\right)\,db=\zeta(2).
	\end{align}
 	We can apply the same process to the edge $\beta_{\overline{12}}$, so there should be two contributions with the opposite sign. Then the corresponding term of the period is then given by summing up the degrees of the edges of $\D$.
 	\begin{align}
 		\int_{\RLC|_\text{codim-2}}\check{\Omega}_z=-\zeta(2)\left(\textcolor{blue}{2(4)+6(3)}+\textcolor{red}{2(-1)}\right)=-24\zeta(2)
 	\end{align}
	\end{example}

	\begin{example}[Error in tropicalization for $\mathcal{F}^{(4)}_3$]
		In four dimensions one could put the defining equation $W_z^\text{trop}$ into the form $x_1+x_2+x_3=b\oplus_z 0$ or
		$x_1+x_2=b_1\oplus_z b_2\oplus_z 0$ along each edge of $\D$, giving contributions of $I^\text{\rom{1}}_{\zeta(3)}$ or $I^\text{\rom{2}}_{\zeta(3)}$ to $\pi(\RLC)$ respectively. Similar to the last example, the edges which lie on the hyperplane $\beta_{\overline{1}}$ defined by the VEX point $(-1,-1,-1,0)$ give a contribution with the opposite sign due the orientation of $\D_\text{ext}$ being flipped. This is also manifest in Equations \ref{TrK1} and \ref{TrK2}, where we can see that  there  be negative contributions to the counts of Type I and Type II affine singularities when $m>2$. Namely, any term that multiplies $(2-m)D_2$ in $K^*$  contribute $-I^\text{\rom{1}}_{\zeta(3)}$ or $-I^\text{\rom{2}}_{\zeta(3)}$. Counting all these up with the appropriate signs, we obtain $\chi=-168$ as before. This calculation of the degrees and counts of the edges of $\D$ is consistent with \cite[Appendix B.2]{BH16}.
	\end{example}

	\section{Conclusions}
	We have shown that one can calculate the B-side period of the mirror real Lagrangian cycle $\pi(\RLC)$ even when the A-side toric ambient space $\mathcal{Y}_\D$ is non-Fano.
	There were several novel features of the calculation from relaxing this ample condition of the canonical bundle $K$.
	The generalized Duistermaat-Heckman measure $\overline{DH}_{\D,\xi}$ was employed to construct the multi-polytope $\D$ from the data of $\Dp$, and we saw that we naturally obtain $\D=(\Dp)^\circ\cup\D_\text{ext}=(\Dp)^\nabla$, realizing the trans-polar construction of Berglund-H\"{u}bsch.
	In a novel algebraic manner, the graded ring $H_*^\text{trop}(\Sigma)$ allowed us to construct $\overline{DH}_{\D,\xi}$ in any number of dimensions.
	Through relating the tropical amoeba of the real Lagrangian cycle $\mathcal{A}^\text{trop}(\RLC)\sim\partial\D$ to the support of the generalized Duistermaat-Heckman measure $\text{supp}\overline{DH}_{\D,\xi}=\D$, we have conjectured that there must be a way to ``see" $\D_\text{ext}$ in the tropical limit.
	
	It is still not known how to use $\pi(\RLC)$ to obtain the higher order $z^k$ terms of the top period $\pi_n$ that we discussed in Section \ref{sec:MS}.
	In Equation \ref{HirzFunPeriod}, one can see there are issues using the traditional method for these non-Fano cases. 
    Given the successes of the Gross-Siebert mirror symmetry program, there should be a way to solve these problems using the tropical methods outlined in this paper.
	By homological mirror symmetry, there are cycles $C$ that are mirror to sheaves $\mathscr{F}$ other than the structure sheaf $\mathscr{O}_X$, so it is also of interest to see how $\pi(C)$ is related to the other $\pi_i$.
	In a forthcoming work \cite{wip}, we plan to similarly analyze periods of ``corrected" cycles $\check{C}^A_z=\{x\in\R_+^n \, | \, W^{A,\text{trop}}(x)=1\}$ using the corrected superpotential defined in Remark \ref{CorrRemark}. This should elucidate unanswered questions about the nature of mirror symmetry for the ambient toric variety itself. Since it is possible for $H^k(\mathcal{Y}_\D,K)\neq 0$ with $k>0$ when $K$ is not ample, we expect higher cohomology contributions to be important for future work in this direction. 
\section*{Acknowledgments}
PB thanks Tristan H\"ubsch for many collaborations and  discussions over the years, especially recently on the topic of non-Fano toric varieties and VEX polytopes. PB would also like to thank Samson Shatashvili for useful discussions on the generalized Duistermaat-Heckman measure, as well as the Hamilton Institute and Mathematics Department at Trinity College Dublin for their hospitality.
ML thanks the lecturers and fellow participants for a very inspirational time at The Physical Mathematics of QFT Summer School 2022 at UMass Amherst. 
ML would also like to thank Giorgi Butbaia for his helpful discussions about Propositions 4.2 and 4.3.
The work of PB and ML was supported by the Department of Energy under grant DE-SC0020220.
ML would also like to acknowledge the UNH graduate school for their  continued support through the 2020, 2021 and 2022 Summer Teaching Assistant Fellowships.

\appendix
\bibliographystyle{JHEP}
\bibliography{ref}

\end{document}